%% file: main.tex
\newif\ifuselm
\uselmtrue
\RemoveFromHook{package/amsthm/after}[firstaid/aliascounter]
\documentclass[a4paper,UKenglish,cleveref, autoref, thm-restate,dvipsnames]{lipics-v2021}

\usepackage{arxiv-lipics-v2021}

\usepackage{todonotes}
\usepackage{xspace}
\usepackage{mathtools}
\usepackage{multicol}
\usepackage{stmaryrd}
\usepackage{amsthm}
\usepackage{amssymb}
\usepackage{nicefrac}
\usepackage{xparse}
\usepackage{array}
\usepackage{makecell}
\usepackage{graphicx}
\usepackage{extarrows}
\usetikzlibrary {automata,positioning} 

\usepackage[bb=boondox,bfbb]{mathalpha}

\usepackage[appendix=strip]{apxproof} % remove proofs
\newtheoremrep{apxtheorem}[theorem]{Theorem}
\newtheoremrep{apxlemma}[theorem]{Lemma}

\nolinenumbers

\hideLIPIcs 
\title{Symmetric two-counter machines%
} 

\author{Krzysztof Jaros}{University of Warsaw}{}
{}
{
Supported by the NCN grant 2024/55/B/ST6/01674.
}
\author{{\L}ukasz~Kami{\'n}ski}{University of Warsaw}{}
{https://orcid.org/0009-0004-1641-9049}
{Supported by the NCN grant 2024/55/B/ST6/01674.}
\author{Maja K\k{a}dzio{\l}ka}{University of Warsaw}{}
{https://orcid.org/0009-0001-1308-9492}
{Supported by the NCN grant 2024/55/B/ST6/01674.}
\author{S{\l}awomir Lasota}{University of Warsaw}{}
{https://orcid.org/0000-0001-8674-4470}
{Partially supported by ERC grant INFSYS, agreement no. 950398, and 
 by the NCN grant 2024/55/B/ST6/01674.}

\authorrunning{K.~Jaros, {\L}.~Kami{\'n}ski, M.~K\k{a}dzio{\l}ka, S.~Lasota} 
\Copyright{Krzysztof~Jaros, Maja~K\k{a}dzio{\l}ka, {\L}ukasz~Kami{\'n}ski, S{\l}awomir~Lasota} 

 %TODO mandatory; please add comma-separated list of keywords

\category{} %optional, e.g. invited paper

\acknowledgements{
We thank Roland Guttenberg
and Vaios Michalakis for 
insightful discussions
on the topic of this research,
as well as Samira Bertsch
and Violet Promeuschel
for their helpful remarks which
helped shape the exposition
in this paper.
}

\input{macros}

\begin{document}

\maketitle

\begin{abstract}
Counter machines are a particularly simple, Turing-complete
model of computation.
We study the expressive power of nondeterministic
2-counter machines which are symmetric, in that their set of
update transitions is invariant under swapping of counters.
Our model is an extension of symmetric 2-\vass, in that
we allow modifying both counters within a single transition.
% On the other hand,
% we impose no restriction on zero-testing transitions.
%mei: skipping this sentence, as it disturbs the flow a bit too much,
%and we end up showing that this choice isn't important
As the main result we demonstrate a dramatic impact of symmetry:
we show that symmetric 2-counter machines
are not Turing-complete, by proving decidability of their 
reachability problem.
We conjecture that the decidability extends to 
machines with 3 or more counters that are symmetric, 
meaning the set of update transitions is
invariant under any permutation of counters. 
We leave this problem open for machines with 3 and more counters,
while identifying a frontier of the impact of symmetry: we prove
undecidability for 3-counter machines invariant under the
cyclic permutation group of order~3.
\end{abstract}

{
\small
\para{Disclosure}
Large language models have been used to assist with proofreading and checking the correctness of the text.
The authors remain fully responsible for the content of the paper, and are confident that its quality
was not negatively impacted by the inherent shortcomings of the technology.
}

\newpage

\input{prefix}
\input{suffix}

\input{variants}

\input{complexity}

%\newpage

\bibliography{bib}

\end{document}

%% file: macros.tex
\newcommand{\saturated}{saturated\xspace}

\newcommand{\Saturated}{Saturated\xspace}
\newcommand{\pre}{\mathrm{pre}}
\newcommand{\post}{\mathrm{post}}
\newcommand{\reachrel}{\Longrightarrow}
\newcommand{\stepp}[1]{\reachrel_{#1}}
\newcommand{\subseteqfin}{\subseteq_\text{\tiny fin}}
\newcommand{\minf}[1]{\textsc{min}(#1)}
\newcommand{\setto}[1]{\{1\ldots #1\}}
\newcommand{\setfromto}[2]{\{#1\ldots #2\}}
\newcommand{\xupdinstr}{\textbf{update}}
\newcommand{\xztinstr}{\textbf{zero-test}}
\newcommand{\updinstr}[1]{\xupdinstr(#1)}
\newcommand{\ztinstr}[1]{\xztinstr(#1)}
\newcommand{\mach}{\mathcal{M}}
\newcommand{\Instr}{\mathcal{I}}

\newcommand{\ie}{i.e.\xspace}
\newcommand{\eg}{e.g.\xspace}

\newcommand{\vass}{\text{\sc vass}\xspace}
\newcommand{\cm}{\text{\sc cm}\xspace}
\newcommand{\figleft}{\textsf{\textbf{Left:}}\xspace}
\newcommand{\figcenter}{\textsf{\textbf{Center:}}\xspace}
\newcommand{\figright}{\textsf{\textbf{Right:}}\xspace}

\newcommand{\cycg}[1]{\mathbb{Z}_{#1}}
\newcommand{\symg}[1]{\textsc{S}_{#1}}

\newcommand{\vr}[1]{\mathbf{#1}}%{\vect{#1}}
\newcommand{\setof}[2]{\{\,#1 \ \mid \ #2\,\}}
\newcommand{\Setof}[2]{\bigl\{\,#1 \ \mid \ #2\,\bigr\}}

\newcommand{\step}{\xlongrightarrow[]{\quad}}
\newcommand{\auxstep}{\dashrightarrow}
\newcommand\xstep[1]{\xrightarrow{#1}}
\newcommand{\Mywlog}{W.l.o.g.\xspace}
\newcommand{\mywlog}{w.l.o.g.\xspace}
\newcommand{\norm}[1]%{\textsc{norm}(#1)}
{\lVert#1\rVert}

\newcommand{\expf}[1]{\exp_{#1}}

\newcommand{\np}{{\sc NP}\xspace}
\newcommand{\nl}{{\sc NL}\xspace}
\newcommand{\pspace}{{\sc PSpace}\xspace}

\newcommand{\para}[1]{\vspace{-2    mm}
\subparagraph*{#1.}}

\newcommand{\Z}{\mathbb{Z}}
\newcommand{\N}{\mathbb{N}}

\newcommand{\V}{\mathcal{V}}

\newcommand{\mysize}[1]{|#1|}

\usetikzlibrary{automata,positioning,shadows}
\usetikzlibrary{calc, shapes, backgrounds,arrows}
\usetikzlibrary{chains, decorations.pathmorphing}
\usetikzlibrary{positioning,decorations.pathreplacing, automata}
\usetikzlibrary{shapes.misc}

\tikzset{every loop/.style={looseness=7}, >=latex}
\tikzset{every picture/.style={>=latex}}
\tikzstyle{PlayerMin}=[draw,circle,minimum size=4mm,inner sep=1.5pt]
\tikzstyle{PlayerMax}=[draw,rectangle,minimum size=7mm,inner sep=1.5pt]
\tikzstyle{Player}=[draw,diamond,minimum size=7mm,inner sep=1.5pt]
\tikzstyle{target}=[circle, minimum size=1mm,inner sep=-2pt]
\tikzstyle{PlayerMinmin}=[draw,circle, minimum size=1.5mm,inner sep=0pt]
\tikzstyle{PlayerMaxmin}=[draw,rectangle,minimum size=1.5mm,inner sep=0pt]
\tikzstyle{PlayerMinsmall}=[draw,circle, minimum size=3mm,inner sep=0pt]
\tikzstyle{PlayerMaxsmall}=[draw,rectangle,minimum size=3mm,inner sep=0pt]
\tikzstyle{leaf}=[draw,diamond,minimum size=7mm,inner sep=1.5pt]

\tikzstyle{strat} =[minimum width=0.1cm,line width=0.01mm,draw=none]
\tikzstyle{vecArrow} = [decoration={markings,mark=at position
1 with {\arrow[scale=.5,>=latex]{>}}}%,postaction={decorate}
]

\theoremstyle{definition}

%% file: prefix.tex
% %\begin{abstract}
% Counter machines are a particularly simple, Turing-complete
% model of computation.
% We study the expressive power of nondeterministic
% counter machines which are symmetric, namely their
% transitions are invariant under swapping of counters.
% As the main result we demonstrate that symmetric 2-counter machines
% are not Turing-complete, by proving decidability of their 
% reachability problem.
% %Our approach builds on specific properties of 2-counter
% %machines, and does not extend to 3 or more counters.
% We conjecture the decidability 
% for nondeterministic machines with arbitrarily many counters,
% as long as (their transitions) are invariant under all permutations
% of counters. 
% We leave the status for symmetric 3-counter machines open, but
% identify a bound on the impact of symmetry by proving
% undecidability for 3-counter machines invariant under the
% cyclic permutation group.
% %\end{abstract}

\bigskip

\section{Introduction}

Counter machines (\cm{s}) are a particularly simple
model of computation \cite{Minsky67}.
A counter machine has a number of registers, called \emph{counters},
each register storing a natural number.
Although the instruction set may vary, 
a common one consists of: counter increment, 
counter decrement, and counter zero-test. 
Many variants of the model have been introduced and studied,
and both deterministic and nondeterministic 
machines are widely studied in the literature 
(often under different names, for instance 
\emph{counter automata}).
Numerous syntactic restrictions of the model are also widely
studied \cite{HKOW09,Ibarra78,HopcroftPansiot,FS00}
since,
as shown already by Minsky in the early 1960s \cite{Minsky61},
in its full generality the model is Turing-complete:
deterministic 2-counter machines can faithfully simulate Turing machines
and have an undecidable halting problem. 

The purpose of this paper is to understand the impact of
\emph{symmetry} on the expressive power of \cm{s}. %counter machines.
We investigate the \emph{nondeterministic} model, where
a machine may choose between different instructions executable
in the same state.

\para{Symmetric counter machines}
Imposing a symmetry condition on the transitions of a counter
machine is inherently sensitive to what can be expressed
within a single instruction.
As such, we specifically choose the set of machine instructions:
counter \emph{updates} that simultaneously modify all counters  
by corresponding deltas, and counter \emph{zero-tests}, which are 
executable only if a given counter has value zero.
Every \emph{transition} $(q,i,q')$ of a machine changes its state
from $q$ to $q'$,
and executes an attached instruction $i$.
Executing a transition in a given configuration constitutes a 
\emph{step} of a machine.
The model we investigate is thus essentially \vass
(vector addition systems with states) \cite{HopcroftPansiot}
extended with zero-tests, and
is clearly equivalent to standard \cm{s}. %counter machines.
In this paper we mostly focus on \cm{s} %counter machines 
that are 
\emph{symmetric}, by which we mean that the set of update transitions
is invariant under all permutations $\sigma$ of counters:
for every update transition $(q,i,q')$, the machine has also a transition
$(q,\sigma(i), q')$, where $\sigma$ acts on an update instruction $i$
by permuting the roles of the counters
(see \Cref{fig:example} for an example).
To stress this point, we choose to impose the invariance only on update transitions,
% (when the instruction $i$ is a counter update), 
but not on zero-test transitions ---
imposing the invariance (symmetry) 
also on zero-tests is not needed for our main 
decidability proof.

\begin{figure}[h]
    \centering
\scalebox{0.6}
{\begin{tikzpicture} [draw=black!70,
  fill=blue!20, node distance = 5cm, on grid, auto,
  every initial by arrow/.style = {thin}]
% State q0 
\node (q0) [state, scale=1.2, fill=blue!20, initial text = {}] {$p$};
% State q1  
\node (q1) [state, scale=1.2, fill=blue!20, right = of q0] {$q$};
% Arrows
\path [-stealth, thick]
    (q0) edge[bend left] node {\small{$\ztinstr 1$}}  (q1)
    (q1) edge[bend left] node {\small$\updinstr{1,1}$}  (q0)
    (q0) edge [loop above] node[above left] {\small$\updinstr{-1, 3}$}()
    (q0) edge [loop below] node[below left] {\small$\updinstr{3, -1}$}()
    (q1) edge [loop below] node[below right] {\small$\updinstr{1, -1}$}()
    (q1) edge [loop above] node[above right] {\small$\updinstr{-1, 1}$}();
\end{tikzpicture}}
\caption{
    An example symmetric 2-counter machine with two control states and
    six transitions, including five symmetric update transitions
    (invariant under the swap of counters) and one zero-test transition.
    Due to the symmetry condition, \eg the $(p, \updinstr{-1, 3}, p)$ transition
    also forces $(p, \updinstr{3, -1}, p)$ to exist. In contrast, we permit
    $\ztinstr 1$ without a corresponding $\ztinstr 2$ transition.
    Also notice the $\updinstr{1,1}$ transition, which is self-symmetric.
}
\label{fig:example}
\end{figure}
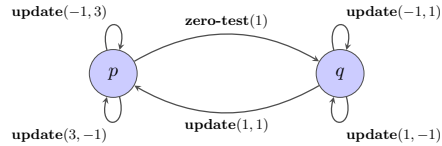

\vspace{-3mm}

\para{Contribution}
The initial motivation of this study is the immense impact of symmetry
on the complexity of the \vass reachability problem:
while the problem is {\sc Ackermann}-complete in general \cite{LS19,CO21,L21},
the complexity drops to
only \pspace-complete for symmetric \vass \cite{KL25}.
As our main result, we confirm the immense drop in case of 
2-\cm{s}
(in this case, symmetry amounts to invariance under all swaps of 
counters).
We prove decidability of the configuration-to-configuration reachability
problem, where one asks if a given machine has a sequence of steps
starting in a given source configuration and ending in
a given target one:

\begin{theorem} \label{thm:main}
The reachability problem is decidable for symmetric 2-counter machines.
\end{theorem}

Symmetric two-counter machines are thus not Turing-complete.
We find this result quite surprising, as it applies to a model
with restricted update instructions but unrestricted zero-test instructions.
% 
% \begin{remark}
% In the proof we do not need symmetry of zero-test transitions,
% only symmetry of updating transitions.
Interestingly, the converse restriction, namely imposing symmetry on
zero-test transitions but not on update transitions, remains 
undecidable.
% \end{remark}

\smallskip

It is natural to consider invariance under subgroups
of permutations: for every subgroup $G\leq \symg d$
one can consider the class of \emph{$G$-counter machines}
($G$-\cm{s}), namely 
$d$-counter machines invariant under all permutations from $G$.
On one end, 
when $G=\symg d$, the general definition instantiates to 
symmetric $d$-\cm{s}.
On the opposite end, when $G$ is the trivial group one gets classical counter
machines.
Between the two extremes there is a landscape of models, 
whose expressive power is related according to the permutation subgroup relation:
whenever $G\leq H$, every $H$-\cm is automatically a
$G$-\cm.
For instance, every $\symg 3$-\cm is automatically
a $\cycg 3$-\cm, for the cyclic 
permutation group $\cycg 3 \leq \symg 3$.
As an additional result, we identify a bound on the impact of symmetry,
already for 3-counter machines, namely we prove undecidability 
for $\cycg 3$-counter machines
(even under an additional assumption that zero-tests are \emph{symmetric}):

\begin{theorem} \label{thm:undecid}
The reachability problem  is undecidable for $\cycg d$-counter machines,
for all $d \geq 3$.
\end{theorem}
Both the above results hold in the hardest case: decidability with
unrestricted zero-tests, and undecidability with symmetric ones. 
This is unsurprising, as in \Cref{sec:variants} we show that assuming
symmetry of zero-tests does not weaken the model.

\para{Open questions}

A fundamental question is whether \Cref{thm:main} extends to larger numbers 
of counters. 
In particular, is the reachability problem still decidable
for symmetric 3-counter machines, whose update transitions are invariant
under all 6 permutations of the 3 counters?
Our decidability proof does not extend to 3 counters, 
as the argument relies on a
number of low-dimensional properties of reachable sets and reachability
relations.
We thus leave this intriguing question for further research, conjecturing
decidability for any number of counters:

\begin{conjecture}
The reachability problem  is decidable for symmetric counter machines.
\end{conjecture}

We focus on decidability, thus
leaving open the exact complexity of the problem for
symmetric 2-\cm{s}.
While we obtain an elementary upper bound for our decision 
procedure,
the only lower bound is
\pspace-hardness inherited from symmetric 2-\vass
% , if updates are encoded in binary 
\cite{KL25}.
It is possible that the reachability problem is \pspace-complete
in symmetric $d$-\cm{s} for every $d\geq 2$.
% \mei{I'm 90\% sure that this should be $d \geq 2$.}

Symmetric \cm{s} where the only update instructions are increments 
or decrements of a single counter, constitute a syntactic fragment of our 
setting. 
It remains to be established the 
expressive power of the fragment, and the complexity of 
its reachability problem.

\para{Outline} 
Once we define the model in \Cref{sec:model}, the subsequent
\Cref{sec:outline,sec:sl,sec:sl2,sec:targetpump} are devoted to the proof of
\Cref{thm:main}.
We proceed in a top-down manner, by consecutive reductions to more and
more basic properties. 
In \Cref{sec:outline} we provide the proof of \Cref{thm:main},
building on effective
semilinearity of certain reachability sets and relations in a
symmetric 2-\cm.
Then, in \Cref{sec:sl,sec:sl2}, we prove the effective semilinearity,
building on the central property of target (resp.~source) pumping
in symmetric 2-\vass, that is subsequently shown in \Cref{sec:targetpump}.
Technically, \Cref{sec:targetpump} is probably the most intricate part 
of the paper, and the only one directly exploiting the symmetry of 2-\cm{s}.
In \Cref{sec:undecid}, we turn to the proof of \Cref{thm:undecid},
followed by \Cref{sec:variants} where we discuss how syntactic restrictions of our model
affect its expressive power.
Finally, in \Cref{sec:complexity}, we provide an elementary upper bound
on the complexity of our decision procedure for symmetric 2-\cm{s}.

\para{Related research}

Identifying decidable restrictions of \cm{s}
is a very active line of research since the 1960s, for instance 
1-counter, reversal-bounded, or restricted 2-counter machines 
\cite{HKOW09,Ibarra78,FS00}.
A prominent example is \cm{s} without zero-tests, that is
vector addition systems with states
(\vass), shown to have decidable reachability in the early 1980s
\cite{Mayr81,Kosaraju82,Lambert92}.
The decidability extends to 
\cm{s}
with only one zero-tested counter \cite{Bonnet11},
or with nested zero-tests \cite{Reinhardt08,GCL25}.
On the other hand, most natural restrictions
preserve Turing-completeness, as long as no constraint is imposed
on zero-tests.
For instance, considering deterministic \cm{s}, 
the reachability problem is still undecidable
in the subclass of \emph{reversible} \cm{s},
where the step relation is reverse deterministic, and hence every step
may be uniquely undone \cite{Morita96}.

In the special case of 1 counter the symmetry requirement is void,
and the reachability problem is \np-complete 
if transition updates are encoded in binary, and \nl-complete 
in unary encoding \cite{HKOW09}.
The known complexity bounds separate symmetric 2-\cm{s} from 1-\cm{s}, 
the former being \pspace-hard while the latter \np-complete,
thus proving that symmetric 2-\cm{s} are not (easily) simulated
by 1-\cm{s}.

Restricting the use of zero-tests in 2-\cm{s} often
leads not only to decidability, but even 
to \pspace-completeness.
Assuming binary encoding,
the \pspace-completeness has been obtained for 2-\vass 
\cite{Blondin15,BEFGHLMT21}
(assuming unary encoding, we have \nl-completeness \cite{ELT16}),
but also for the further restriction to symmetric 2-\vass \cite{KL25},
and for a relaxation of 2-\vass to 
2-\cm{s} with just one zero-tested counter \cite{FLS18}.
In all these cases 
the reachability relation is effectively semilinear, which
is the key for obtaining a \pspace upper bound, but which 
fails already for 3-\vass 
\cite{HopcroftPansiot}.
We do not know the exact complexity of symmetric 2-\cm{s}.
% but it is possible that the complexity remains \pspace.
Neither do we know if the reachability relation of a
symmetric 2-\cm is semilinear, but we believe it is so.

Contrary to the above-mentioned restrictions of zero-tests, 
our model leaves the zero-tests unrestricted,
which rarely results in a decidable model
--- the primary example being a result of Dudenhefner
\cite{Dudenhefner22}, showing that
% which makes \Cref{thm:main} particularly interesting.
% We conclude by recalling another recently obtained decidability result that
% also does not assume any explicit restriction of zero-tests 
% similarly to \cite{Dudenhefner22} showing that
% , and hence seems as surprising as our \Cref{thm:main}.
% It amounts to observing that
% the above-mentioned Turing-completeness of
% reversible counter machines seems to be very fragile with respect
% to the choice of instruction set of a \cm.
% As shown in \cite{Dudenhefner22},
for certain 
choice of instruction set, the reversible 2-\cm{s}
become decidable, in contrast with \cite{Morita96}.

\section{Symmetric counter machines} \label{sec:model}

Let $d$ be a positive integer.
A \emph{$d$-counter machine} ($d$-\cm in short) 
$\mach$
is a finite-state automaton whose transitions
may perform two kinds of \emph{instructions},
namely \emph{update} and \emph{zero-test} instructions,
on $d$ \emph{counters}, namely nonnegative integer values
stored as part of the configuration.
Formally, we define such a machine as 
$\mach = (Q,T)$, where $Q$ is a finite set of control states,
$T\subseteqfin Q\times\Instr\times Q$ is a finite set of 
\emph{transitions} of $\mach$,
and the set $\Instr$ of \emph{instructions} is given by:
\[
\Instr \ = \ \setof{\updinstr{\vr w}}{\vr w\in\Z^d} \ \cup \ 
\setof{\ztinstr{i}}{i\in \setto d}.
\]
% \mei{it is ambiguous here whether $i \in d$ is supposed to mean
% $\{0, ..., d-1\}$ or $\{1, ..., d\}$. We use the more explicit notation
% below, maybe we should do so here as well?}
Execution of $\updinstr{\vr w}$ amounts to adding
the vector
$\vr w$ to the counters, and execution of 
$\ztinstr{i}$ amounts to checking that counter $i$
has value 0.
Formally, we define the semantics of $\mach$ by considering
its \emph{configurations} $(q,\vr v) \in Q\times \N^d$, written $q(\vr v)$, and
the step relation between configurations determined by 
transitions of $\mach$:
$q(\vr v) \step q'(\vr v')$ if 
\begin{itemize}
\item 
$\vr v' = \vr v + \vr w$,  for some  $(q,\updinstr{\vr w},q')\in T$, or
\item  
$\vr v' = \vr v$ and $\vr v[i]=0$, for some 
$(q,\ztinstr i, q')\in T$,
\end{itemize}
where by $\vr v[i]$ we denote the $i$-th coordinate of the vector $\vr v$.
% \mei{I chose to change the notation to $\vr v[i]$ to make the below
% definition of $\sigma(\vr w)$ read better; I'm open to discussion on this.
% One potential sticking point is that the $q[r]$ notation we later use in the
% 1-cm construction looks the same, but I'd say the contexts are sufficiently different.}

We use the term $d$-\vass (vector addition system with states)
for a $d$-\cm without zero-test instructions.
The \emph{underlying} \vass of $\mach$ is
obtained by removing from $\mach$ all $\xztinstr$ transitions.

The symmetric group $\symg d$ acts naturally on vectors in $\Z^d$
and hence also on update transitions:
for $\sigma\in \symg d$ and $t=(q,\updinstr{\vr w},q')\in T$, we define
$\sigma(t)=(q,\updinstr{\sigma(\vr w)},q')$, where $\sigma(\vr w)[\sigma(i)]=\vr w[i]$
for every $i\in\setto d$.
A $d$-\cm $\mach$ is \emph{symmetric} if its set of update transitions
is invariant under the natural action: for every $\sigma\in \symg d$ 
and $(q,\updinstr{\vr w},q')\in T$, we have $(q,\updinstr {\sigma (\vr w)},q')\in T$.
Importantly, 
symmetry is imposed on $\xupdinstr$ transitions, but not on 
$\xztinstr$ transitions of $\mach$.
% We assume, \mywlog, that every zero-test applies to a single counter.
% 
In the sequel
we mostly focus on symmetric 2-\cm{s},
namely 2-\cm{s} whose set of $\xupdinstr$ transitions is invariant under the
swap of counters.

\begin{remark}
    Likewise one defines $G$-\cm{s}, counter machines whose 
    set of $\xupdinstr$ 
    transitions is invariant under the action of a permutation subgroup
    $G\leq\symg d$.
    Note that this notion is contravariant: when $G$ is a permutation subgroup of $H$, every
    $H$-\cm is automatically a $G$-\cm, and not vice versa.
    We investigate this more general model only in 
    \Cref{sec:undecid,sec:variants}. 
\end{remark}
We define the \emph{reachability relation} (between configurations)
as the reflexive-transitive closure $\step^*$ of $\step$.
The \emph{reachability problem} asks, given a \cm $\mach$ and two
of its configurations, a source $s\in Q\times\N^d$ and a target $t\in Q\times\N^d$, 
whether 
$s \step^* t$
(in which case we say that $t$ is \emph{reachable} from $s$).
A reachability certificate, namely any sequence $\rho$ of configurations
\[
q_0(\vr v_0) \step q_1(\vr v_1) \step \cdots \step q_n(\vr v_n)
\]
with $s = q_0(\vr v_0)$ and $t = q_n(\vr v_n)$,
% \mei{The kerning here suggests the bracketing
% $c = (c_0 \step \cdots) = c'$, which becomes obviously wrong only once
% you try to actually parse the equation.
% Since the additional spacing was put in explicitly, I'd like to know
% why before fixing this.}
related by the step relation,
is said to be a \emph{run} from $s$ to $t$.

\para{Variants}
Two natural restrictions of the model are obtained by requiring that:
% \begin{itemize}
(1) the set of $\xztinstr$ transitions is also \emph{symmetric}, 
that is invariant under permutations of counters; 
(2) all $\xupdinstr$ transitions $(q,\updinstr{\vr w},q')$ are 
\emph{single-counter}, that is $\vr w$ is nonzero on at most one coordinate.
% \end{itemize}
% 
As we discuss later in \Cref{sec:variants}, the first 
restriction is equally powerful as the unrestricted model,
and imposing both restrictions makes the model essentially degenerate 
to 1-\vass, once the source and target vectors are zero.
We do not know consequences of 
imposing solely the second restriction.

% A 2-\vass $\V$ is symmetric if whenever it has a transition
% $q \xstep{(x_1,x_2)} q'$
% then it also has the transposed transition 
% $q \xstep{(x_2,x_1)} q'$.

% %
% A \emph{symmetric counter machine} with two counters (\sym 2-\cm) is a 
% counter machine whose \vass component (\ie, the \vass obtained by removing zero tests) is a symmetric \vass.
% Thus, transitions are symmetric,
% but zero tests need not be, and can distinguish the two counters.

% We prove that reachability in \sym 2-\cm is decidable.
% The proof consists of two independent sections:
% \cref{sec:targetpump}, which proves a general lemma about \sym 2-\cm{s},
% and \cref{sec:strategy,sec:proof}, which use that lemma (as a black-box)
% to prove the decidability result.

\section{Outline of decidability proof}\label{sec:outline}

% % As above, we use $\step$ to denote step relation of a 2-\vass, namely without zero tests.
% First some notational preprocessing.

We now outline the proof of \Cref{thm:main}.
To this end we
fix a given symmetric 2-\cm $\mach = (Q,T)$ and 
two configurations: a source $s$, and a target $t$ ($s,t \in Q\times\N^2$).

\para{Notations}
For relations $R \subseteq A \times B$,
$R' \subseteq B \times C$, we denote by
$R \circ R' \subseteq A \times C$
the diagrammatic-order composition of the relations.
In particular, whenever convenient we consider sets
$S \subseteq A$, $T \subseteq B$ as
relations from/to the singleton set, denoting
by $S \circ R \subseteq B$ and $R \circ T \subseteq C$ the (inverse) image of a set
under a relation $R$.
The reflexive transitive closure of $R\subseteq A\times A$ 
is denoted by $R^*$, as usual.

\para{Decomposing a run into segments}
For convenience, we perform on $\mach$ the following notational
preprocessing:
we replace every zero-test transition
$t=(q,\ztinstr i, q')\in T$ by a sequence of two transitions,
\[
(q, \ztinstr i, q_t) \qquad
(q_t, \updinstr {0,0}, q'),
\]
% q \xstep{\textsf{c}_i=0?} S_i^t \xstep{(0,0)} q',\]
where $q_t$ is a new state, unique to the transition $t$.
Notice that reachability in the new symmetric 2-\cm thus obtained 
is equivalent to the original one, 
and that it has special states $q_t$
which can only be entered if the counter zero-tested by 
transition $t$ is zero;
we call them \emph{zero-testing states}.
The zero-testing states have the property
that they ``remember'' which counter is zero 
when the control of the machine is at them,
and that zero-test transitions are exactly those entering
a zero-testing state.
A fragment (infix) of a run between two consecutive zero-testing
states consists of only update transitions, except for the very last
transition which is a zero-test.
Such a fragment, called \emph{segment} in the sequel,
may be considered as a 2-\vass run
that ends with one counter 0, and which counter is 0 is
determined by the ending, zero-testing state.
Formally, we will consider a segment to be a run in the
underlying symmetric 2-\vass $\V$,
leaving the final zero-test implicit in the fact
that the final configuration of a segment has the relevant
counter be zero.
% extended with a single zero-testing
% transition to a chosen ending zero-testing state 
% with no outgoing transitions.
% obtained by removing all zero-testing states except the starting and ending one.
%  transitions $(q, \ztinstr i, q^z)$ by
% the corresponding $(0,0)$ update one
% $(q, \updinstr{0,0}, q^z)$.

\begin{remark} \label{rem:wlog}
Furthermore, \mywlog we may assume that one counter is 0
in both source and target configurations.
    For instance, if the source configuration $q(v_1, v_2)$ satisfies $0<v_1\leq v_2$,
    we replace it by $q_s(0,v_2-v_1)$ for a freshly added source state $q_s$,
    and add an additional transition $(q_s, \updinstr{v_1, v_1}, q)$.
We thus assume, \mywlog, 
that the source and target states are zero-testing.
\end{remark}

%into those corresponding to counter $\textsf{c}_1$ and to counter $\textsf{c}_2$,
%so we tell which counter is zero by looking at the state.

% The strategy of the proof is to 
In the decidability proof we exploit the unique
decomposition of a run $\rho$ from source $s$ to
target $t$ 
into segments:
%\[ \rho=\rho_0 z_1 \rho_1 z_2 \dots \rho_{n-1}z_n\rho_n,\]
\[ \rho=\rho_0 \rho_1 \dots \rho_n,\]
where 
each $\rho_i$ starts and ends at consecutive zero-testing states of 
$\rho$.
% ---%
% we shall call the $\rho_i$ the \emph{\vass segments} of $\rho$.
%
Every $\rho_i$ is thus
a run in the underlying 2-\vass $\V$,
having one of the following forms:
\begin{alignat}{2} 
  \begin{aligned} \label{eq:4cases}
q(a,0) &\step^\ast q'(0,a') \qquad\qquad  & q(0,a) &\step^\ast q'(a',0)\\
q(0,a) &\step^\ast q'(0,a') \qquad\qquad  & q(a,0) &\step^\ast q'(a',0),
% \\
% q(0,a) &\step^\ast q'(a',0)\\
% q(a,0) &\step^\ast q'(a',0),
  \end{aligned}
\end{alignat}
where the initial and final states $q$ and $q'$ are zero-testing states, and no zero-testing
state is visited in between.
Notice that we may write just
% \begin{equation*} %\label{eq:basictran}
$q(a) \step^* q'(a')$
% \end{equation*}
without specifying which of the four cases occurs,
since the (zero-testing) states $q$, $q'$ uniquely determine that.
We follow the notational shorthand %\eqref{eq:basictran} 
throughout and, instead of 
$\step^*$, which is a binary relation on $Q\times\N^2$,
%  \  \subseteq  (Q\times\N^2)^2$
we prefer to use 
a binary relation on $Q\times\N$, 
\[
\reachrel \ \ \subseteq \ (Q\times\N)^2,
\]
where $q(a)\reachrel q'(a')$ if and only if 
any of the four cases of \eqref{eq:4cases} holds, namely
the underlying 2-\vass $\V$ has a corresponding run that does not
visit zero-testing states except for the initial and final ones.
% Without loss of generality, we also assume 
% that both the source and target vector have at least one 0.
In view of Remark \ref{rem:wlog}, our task is to decide if 
$s\reachrel^* t$
for $s,t\in Q\times\N$,
that is whether a given pair $(s,t)$ belongs to the reflexive-transitive
closure of $\reachrel$.

\medskip

\begin{figure}[h]
\centering
\includegraphics[scale=0.7]{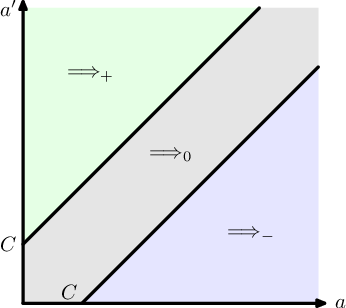}
\caption{The three relations $\stepp +$, $\stepp -$, and $\stepp 0$.
The axes do \emph{not} correspond to the counters, 
but rather the initial and final values $a$ and $a'$ of a segment.
}
\label{fig:areas}
\end{figure}

\vspace{-2mm}

\para{Three types of segments}

As we shall soon see, the relation $\reachrel$
can be split into three disjoint areas with different properties,
delineated by a machine-dependent constant~$C$
(see \Cref{fig:areas}).
Thus, every segment $\rho_i \colon q(a) \reachrel q'(a')$, 
depending on the value of $a-a'$,
falls into one of three categories,
%split the relation $\step^\ast$ into three subrelations,
%which represent different kinds of \vass runs in $(Q\times\N^2)$:
denoted by $\stepp +$, $\stepp -$, and $\stepp 0$.
For a segment $q(a) \reachrel q'(a')$, we write:
\begin{equation}
% \left.
\begin{aligned}\label{eq:tran0-+}
& q(a) \stepp + q'(a'),\quad && 
\text{if\ \ $a' > a\,+C${}};\\
& q(a) \stepp - q'(a'),\quad && 
\text{if\ \ $a\, > a'+C${}};\\
& q(a) \stepp 0\, q'(a'),\quad && 
\text{if\ \ $|a-a'|\leq C $},
\end{aligned}
% \right\}
\end{equation}
where the precise definition of the constant $C$ will be given in \Cref{sec:targetpump}.
The reachability relation can now be expressed as
\[
\reachrel^* \ = \ 
\big({\stepp 0} \ \cup \ {\stepp -} \,\cup \ {\stepp +} \big)^*,
\] 
an arbitrary composition
of the three sub-relations arbitrarily many times
(note that the composition assures that \emph{control states match}).

Using this categorization 
we split the run $\rho$ into three consecutive \emph{phases} (each possibly empty): 

\begin{itemize}
\item $\Phi_1$, the \emph{prefix} of $\rho$ ending with the \emph{last} $\rho_i$ of the form $\stepp +$,
\item $\Phi_2$, an \emph{infix} of $\rho$ consisting only of $\rho_i$'s of the form $\stepp {0}$, 
\item $\Phi_3$, the remaining \emph{suffix} of $\rho$ starting 
with \emph{the first remaining} $\rho_i$ of the form $\stepp -$.
% (followed by other \vass segments of the forms $\stepp -$ and $\stepp 0$).
\end{itemize}

\noindent
Clearly, the composition of the three phases yields
the reachability relation, that is
\begin{align} \label{eq:reachrel}
{\reachrel^*} \ = \ \Phi_1 \circ \Phi_2 \circ \Phi_3,
\end{align}
for  the three relations defined by the following rational expressions:
\begin{align}
  \begin{aligned} \label{eq:regexp}
      &\Phi_1 \ = \ (({\stepp 0} \ \cup \ {\stepp -})^\ast \circ\, {\stepp +})^\ast,\\
      &\Phi_2 \ = \ {\stepp 0^\ast},\\
      &\Phi_3 \ = \ ({\stepp -} \circ\, {\stepp{0}^\ast})^\ast.
  \end{aligned}
\end{align}

% \noindent
% Above, $\circ$ stands for composition of relations, and $^*$ for 
% composition iteration, namely $R^* = $.

\para{Proof of \Cref{thm:main}}

The proof uses \emph{semilinear sets}.
Any subset of $\N^d$ of the form
\begin{align} \label{eq:linset}
% \lin b {\{p_1,\dots,p_n\}} = 
\setof{b + k_1 p_1 + \cdots + k_n p_n}{k_1, \dots, k_n \in \N},
\end{align}
for $b, p_1, \dots, p_n \in \N^d$,
is called a
\emph{linear} set with base $b$ and periods $p_1, \dots, p_n$.
A \emph{semilinear} set is a finite union of linear sets in the same
dimension $d$.
We say $X \subseteq Q \times \N^d$ is semilinear 
if for every $q \in Q$, 
the set 
\(\{n \mid (q,n)\in X\} \subseteq \N^d\)
is semilinear.
Similarly, 
$R \subseteq (Q\times \N^d)^2$ is semilinear if for every $q,q' \in Q$,
the set
\( \{(n,n') \mid ((q,n),(q',n')) \in R \} \subseteq \N^{2d}\)
is semilinear.

\Cref{thm:main} is proved by showing that the relation $\Phi_2$, 
as well as
the image of relation $\Phi_1$ and inverse image of $\Phi_3$,
are all
\emph{effectively semilinear}, \ie they are 
semilinear sets whose descriptions are computable
given the underlying 2-\cm $\mach$ and source and target
$s, t\in Q\times\N$.
Importantly, we only need to consider semilinear sets and relations
in dimension $d=1$.

\begin{lemma} \label{lem:0star}
Given a 2-\cm $\mach$, the relation
$\Phi_2 \subseteq (Q \times \N)^2$ 
is effectively semilinear.
% The relation $\stepp {0}^*$ is effectively semilinear.
\end{lemma}

% \begin{lemma} \label{lem:compbackfor}
% One can compute 
% the set 
% $\big(\stepp - \ \circ \ \stepp {0}^*\big)^* \ \circ \ \{t\}$.
% \end{lemma}

\begin{lemma} \label{lem:compbackfor}
Given a symmetric 2-\cm $\mach$ and source $s$ and target $t$, the sets
\begin{alignat*}{5}
    \Phi_3 \circ \{t\} \ &:= \  \setof{y}{(y,t) &&\in \Phi_3} \ &&\subseteq \ Q \times \N, \\
    \{s\} \circ \Phi_1 \ &:= \  \setof{x}{(s,x) &&\in \Phi_1} \ &&\subseteq \ Q \times \N
\end{alignat*}
are effectively semilinear.
% One can compute 
% the set 
% $\{s\} \circ \big(\stepp {0}^* \ \circ \ E \ \circ \ \stepp +)^*$.
\end{lemma}

%  $\Phi_2 \subseteq (Q \times \N)^2$ 
% is an \emph{effectively semilinear} relation
% %\footnote{This means $\{n \in \N \colon (q,n) \in \Phi_3 \}$ is effectively semilinear, for every $q\in Q$.}%
% % 
% and that for given source $s$ and target $t$, the sets
% \begin{alignat*}{6}
% &\Phi_3 \circ \{t\} \ := \  \{y &&\mid (y,t) \in \Phi_3 \hskip 0.5pt\}  \ \subseteq \ Q \times \N, \\
% &\{s\} \circ \Phi_1 \ := \  \{x &&\mid (s,x) \in \Phi_1 \} \ \subseteq \ Q \times \N
% \end{alignat*}
% are also \emph{effectively semilinear}.

\begin{remark}
The proof of \Cref{lem:0star} does not use symmetry of a 2-\cm
and the lemma holds for \emph{all} 2-\cm{s}, as opposed to
the proof of \Cref{lem:compbackfor} that crucially uses the symmetry assumption.
\end{remark}

Relying on the effective 
semilinearity given in \Cref{lem:0star,lem:compbackfor},
% and on the effective equivalence between semilinear sets and Presburger 
% Arithmetic \cite{GS66},
the proof of \Cref{thm:main} is immediate:
the reachability question $s\reachrel^* t$
is, according to \eqref{eq:reachrel}, equivalent to
%\( (s,t) \in \Phi_1\circ\Phi_2\circ\Phi_3,\) \ie 
% \begin{equation}
%     \exists\, x,y \in Q\times\N \ : \  x \in (\{s\} \circ \Phi_1) \ \wedge \ (x,y) \in \Phi_2 
%     \ \wedge \ y\in (\Phi_3 \circ \{t\});
% \end{equation}
the nonemptiness of the intersection
\[ \bigl(\{s\} \circ \Phi_1\bigr) \times 
\bigl(\Phi_3\circ \{t\} \bigr)
\quad \cap \quad
\Phi_2
\]
of effectively semilinear sets,
which is again effectively semilinear \cite[Thm.~6]{taming}.
% effectively 
% expressible in Presburger Arithmetic, and hence decidable 
% \cite{Presburger30,Oppen78}.
% \lukasz{Strictly speaking, using Presburger arithmetic increases the complexity to sometheing 
% like $7ExpSpace$. 
% Maybe it is better to write about intersection of semilinear sets which is cheap.}
It thus remains to demonstrate the two key lemmas.
Continuing in a top-down manner, we first focus on
\Cref{lem:compbackfor}, whose proof makes use of
\Cref{lem:0star}.

%Notationally, we allows ourselves to compose a relation $R\subseteq (Q\times\N)^2$ 
%with a set $X\subseteq Q\times\N$ to denote the set inverse image $R\circ X\subseteq Q\times\N$, 
%or the image $X\circ R \subseteq Q\times\N$.
%% Let $I_s = \{(s,s)\}$ be the relation containing exactly one pair, where $s$ is the source configuration.
%% Likewise, let $I_t = \{(t,t)\}$.
%We will be able to decide if $s\step^* t$, given source and target configurations $s,t\in Q\times\N$,
%as long as we prove the two following fact: one can (backward) compute the inverse image
%$E\circ \{t\}$, and (forward) compute the image 
%$\{s\} \circ \big(\stepp {0}^* \ \circ \ E \ \circ \ \stepp +)^*$,
%according to the partition into phases.

\section{Effective semilinearity for $\Phi_1$ and $\Phi_3$: 
Proof of \Cref{lem:compbackfor}} \label{sec:sl}

% \para{Target and source pumping}

The \emph{target pumping} 
property of segments of the form $\stepp +$ is 
a central technical tool in the proof of
\Cref{lem:compbackfor}:

\newcommand{\punkt}[2]{({#1}{\rightarrow}{#2})}
\begin{lemma} [Target pumping] \label{lem:targetpump}
Given a symmetric 2-\vass $\V$,
one can compute numbers $C, D >0$ 
such that for every pair of control states $q, q'$ of $\V$ and  
$a, a'\in\N$ such that $a' > a + C$,%
% \[
% \punkt 1 2 \qquad
% \text{if} \quad q(a,0)\step^* q'(0,a') \quad \text{then} \quad q(a,0) \step^* q'(0, a'+D),
% \]
% and similarly for the other three symmetric cases:
\begin{alignat*}{4}
&\punkt 1 2 \qquad
\text{if} \quad q(a,0)\step^* q'(0,a') \quad &&\text{then} \quad q(a,0) \step^* q'(0, a'+D),\\
&\punkt 2 2 \qquad
\text{if} \quad q(0,a)\step^* q'(0,a') \quad &&\text{then}  \quad  q(0,a) \step^* q'(0, a'+D),\\
&\punkt 2 1 \qquad
\text{if} \quad q(0,a)\step^* q'(a',0) \quad &&\text{then}  \quad  q(0,a) \step^* q'(a'+D, 0),\\
&\punkt 1 1 \qquad
\text{if} \quad q(a,0)\step^* q'(a',0) \quad &&\text{then}  \quad  q(a,0) \step^* q'(a'+D,0).
\end{alignat*}
\end{lemma}

\begin{remark}
% It is worth noting that 
The proof of \Cref{lem:targetpump},
to be given in \Cref{sec:targetpump}, is the (only) part of
our development that unavoidably relies on symmetry.
\Cref{lem:targetpump} fails for non-symmetric 2-\vass,
for instance point $\punkt 1 2$ fails 
for a 2-\vass with one state $q$ and one transition $(q,(-1,2),q)$.
% \Cref{lem:targetpump} is a cornerstone of the proof
% of \Cref{lem:compbackfor}.
% \mei{We already said this just above the statement
% of the lemma, let's not repeat ourselves}
\end{remark}

We apply \Cref{lem:targetpump} to the underlying 
symmetric 2-\vass $\V$ of a given symmetric 2-\cm $\mach$.
The four claims of \Cref{lem:targetpump} jointly imply that 
$\stepp +$ is \emph{target-pumpable}, 
that is for every pair of zero-testing control states $q, q'$,
% and  $a, a'\in\N$ such that $a, a'-a > C$, 
\begin{align} \label{eq:tpump}
q(a)\stepp + q'(a') \quad \implies  \quad  q(a) \stepp + q'(a'+D).
\end{align}
The reverse version of the four claims of \Cref{lem:targetpump}, 
where we swap source and target and reverse all transitions
(the reverse of $(q,\vr w,q')$ is $(q',-\vr w,q)$, as expected), 
implies that $\stepp -$ is \emph{source-pumpable}, that is 
\begin{align} \label{eq:spump}
q(a)\stepp - q'(a') \quad \implies  \quad  q(a+D) \stepp - q'(a').
\end{align}
Note that we can assume the same $D$ in \eqref{eq:tpump} and 
\eqref{eq:spump}, by taking the least common multiple of the
two constants given by the two applications of \Cref{lem:targetpump},
and also the same $C$, by taking the maximum of the two constants.
In the sequel, fix the constant 
$D$ from \eqref{eq:tpump} and \eqref{eq:spump}.

\para{\Saturated sets}

It is folklore that semilinear sets $X\subseteq Q\times\N$
are exactly the \emph{ultimately periodic} sets, \ie those satisfying
$q(a)\in X\iff q(a+N) \in X$ for some $N>0,M\geq 0$, all $q\in Q$ 
and all sufficiently large $a\geq M$.
We call $N$ and $M$ a \emph{period} and \emph{threshold} of $X$.

Given $N>0$, we define $N$-\emph{\saturated} subsets $X\subseteq Q\times\N$ as those
where for all $(q,a) \in X$, also $(q,a+N) \in X$.
% if $(q,a) \in X \implies (q,a+N) \in X$ for all $q\in Q$ and $a\in\N$.
$N$-\saturated sets
are ultimately periodic, and hence semilinear.
An $N$-\saturated set $X$ is representable by a function 
$\minf X : Q\times \setfromto 0 {N-1} \to \N \cup \{\infty\}$ 
that maps a state $q\in Q$ and remainder $r<N$ to the least number $n$ 
congruent to $r$ such that $(q,n)\in X$,
\[
    \minf X (q,r) \ = \ \min\ \setof{n}{(q,n) \in X, \ \ n \equiv r \bmod N},
\]
and to $\infty$ if there is no such number $n$.
We call $N$ a \emph{period} of $X$, and the numbers $n < \infty$
appearing in the image of $\minf X$ its \emph{bases}.
For instance, in consequence of \eqref{eq:tpump} and \eqref{eq:spump}, 
the image $S\circ {\stepp +}$ as well as the inverse image
${\stepp -} \circ\, S$ is always $D$-\saturated, regardless of
$S\subseteq Q\times\N$.

For termination of our decision procedure it will be crucial that
the semilinear sets it maintains form
non-decreasing chains 
of $D$-\saturated sets $X\subseteq Q\times\N$, since we have:
%  the following lemma:
\begin{lemma} \label{lem:stab}
    For every $N>0$, 
    every non-decreasing chain of $X_1 \subseteq X_2 \subseteq \cdots$ of $N$-\saturated subsets of $Q\times\N$ stabilizes.
\end{lemma}
\begin{proof}
The sequence $\minf {X_1}, \minf{X_2}, \ldots$ is a non-increasing
chain of functions, ordered pointwise 
(assuming $n<\infty$ for $n\in\N$), 
and hence necessarily stabilizes due to well-foundedness
of the function space $Q\times \setfromto 0 {N-1} \to \N \cup \{\infty\}$.
\end{proof}

Before the proofs of \Cref{lem:0star,lem:compbackfor},
a basic observation concerning individual segments:

\begin{lemma} \label{lem:0-+}
The relations $\stepp 0, \ \stepp -, \ \stepp +$ are effectively semilinear.
\end{lemma}
\begin{proof}
The reachability relation of the underlying 2-\vass $\V$
is effectively semilinear \cite{Blondin15},
which implies that so is $\reachrel$, the restriction thereof 
enforcing a counter to be 0 initially or finally, 
and forbidding zero-testing states except the source and target ones.
In consequence, its further restrictions 
$\stepp 0, \ \stepp -, \ \stepp +$, each of them
obtained by intersecting with a semilinear
set, are also effectively semilinear.
\end{proof}

\begin{proof}[Proof of \Cref{lem:compbackfor}]
Recall that $\Phi_3 \ = \ ({\stepp -} \circ\, {\stepp{0}^\ast})^\ast$.
By \Cref{lem:0-+,lem:0star}, both the relations $\stepp -$ and
$\stepp{0}^*$ are effectively semilinear.
By a standard argument building upon \cite[Thm.~6]{taming},
% Using the effective equivalence between semilinear sets and Presburger Arithmetic,
the composed relation 
${\stepp -} \circ\, {\stepp{0}^\ast}$ is also effectively semilinear,
in particular meaning that
the set
$({\stepp -} \circ\, {\stepp{0}^\ast}) \,\circ\, S$ is semilinear for every given 
semilinear set $S$.
Furthermore, as $\stepp -$ is source pumpable \eqref{eq:spump},
and therefore the composed relation ${\stepp -} \circ\, {\stepp {0}^\ast}$
is so as well, the set
$({\stepp -} \circ\, {\stepp{0}^\ast}) \,\circ\, S$ is always $D$-\saturated.

In order to  argue that the set $\Phi_3 \circ \{t\}$
is effectively semilinear we make use of \Cref{lem:stab},
iterating the transformation
\[
    X_{i+1} \ := \ X_{i} \ \cup \ ({\stepp -} \,\circ\  {\stepp{0}^\ast}) \ \circ\ \! X_{i},
\]
starting from the initial set
\[
    X_0 \ := \ ({\stepp -} \,\circ\ {\stepp{0}^\ast}) \ \circ\ \{t\},
\]
until a fixpoint is reached. By the above observations, we deduce that
both the initial set $X_0$, as well as every successive iterate $X_{i+1}$,
are effectively semilinear and $D$-\saturated; moreover it is
self-evident that the sequence $X_0 \subseteq X_1 \subseteq \cdots$
is non-decreasing.

More explicitly, by \Cref{lem:stab} we have that at starting at some $j$,
\[
X_j \ = \ X_{j+1} \ = \ X_{j+2} \ = \ \cdots
\] 
Since both the iteration
$X_i \mapsto X_{i+1}$ and the equality test $X_i = X_{i+1}$ are effective,
the set 
\[
\Phi_3\circ \{t\} \ = \ \{t\} \ \cup \ \bigcup_{i\in\N} X_i \ = \ 
\{t\} \ \cup \ X_j
\]
is effectively semilinear.
We conclude the first part by 
noting that the proof is not specific to the starting set $\{t\}$,
and shows that $\Phi_3 \circ T$ is effectively 
semilinear given any semilinear $T\subseteq Q\times\N$,
which will be useful in the second part.

\medskip

Now we proceed to arguing that
the set $\{s\} \circ \Phi_1$ is effectively semilinear.
To this aim we prefer to use another rational expression defining 
$\Phi_1$, where $\Phi_3$ appears as a subexpression:
\[
    \Phi_1 \ = \ ({\stepp{0}^*} \ \circ \ \Phi_3 \ \circ \ {\stepp +})^\ast;
% \qquad\qquad\qquad
% \text{recalling } \Phi_3 \ = \ (\stepp - \circ \stepp {0}^\ast)^\ast;
\]
indeed, as in \eqref{eq:regexp} in \Cref{sec:outline}, the inner subexpression ${\stepp{0}^*} \,\circ\, \Phi_3 \,\circ\, {\stepp +}$ 
describes an arbitrary
composition of segments of type $\stepp 0$ and $\stepp -$ that ends 
with one segment of type $\stepp +$.

The structure of the proof is similar to the previous part, namely
we proceed by iterating the relation defined by the inner subexpression, 
that is we consider the initial set
and the consecutive iterates:
\[
X_0 \ := \ f\bigl(\{s\}\bigr) 
\qquad\qquad\qquad
% \ \circ \ \bigl((\stepp 0)^* \ \circ \ \Phi_3 \ \circ \ \stepp +\bigr)
% \]
% \[
X_{i+1} \ := \ X_i \ \ \cup \ \ 
f\bigl(X_i\bigr)
%  \  \circ \  \bigl((\stepp 0)^* \ \circ \ \Phi_3 \ \circ \ \stepp +\bigr).
\]
where
\[
f(S) \ = \ S \ 
\circ \ \bigl({\stepp{0}^*} \ \circ \ \Phi_3 \ \circ \ {\stepp +}\bigr).
\]
Similarly as before, knowing that $\stepp +$ is target pumpable, we know
that all $X_i$ are $D$-\saturated sets, and therefore the sequence 
$X_0, X_1, \ldots$ stabilizes, and hence
\[
\{s\} \,\circ\, \Phi_1 \ = \ 
 \{s\} \ \cup \  X_j
\]
for some $j$, similarly as before.
The only missing ingredient is computability of the iteration
$X_i \mapsto X_{i+1}$.
However, in order to apply the previously shown computability of 
the function
\begin{align} \label{eq:back}
T \ \mapsto \ \Phi_3 \ \circ \ T
\end{align}
which proceeds 'backward', for computability
of the function
$
S \, \mapsto  f(S)
%  \ := \ S \ 
% \circ \ \bigl((\stepp 0)^* \ \circ \ \Phi_3 \ \circ \ \stepp +\bigr)
$
which proceeds 'forward', we need a trick.
Note that the relation 
${\stepp{0}^*} \  \circ \ \Phi_3 \ \circ \  {\stepp +}$ is always
target-pumpable since $\stepp +$ is so by \eqref{eq:tpump},
and hence $f(S)$ is always $D$-\saturated
regardless of $S$.
To compute $f(S)$, one first establishes, for every state $q\in Q$, 
which remainders $r$ modulo $D$ appear in this set,
by applying the transformation \eqref{eq:back} to the set
\begin{align} \label{eq:set1}
\bigl(\stepp + \circ \ X_{q,r}\bigr)
,
\end{align} 
% where $M_r = r + D\cdot\N$ contains all numbers congruent to $r$ modulo $D$, 
where $X_{q,r} = \setof{q(n)}{n\equiv r \bmod D}$,
and checking nonemptiness of intersection with 
the set 
\begin{align} \label{eq:set2}
S\ \,\circ \ \stepp{0}^*.
\end{align} 
Both the sets \eqref{eq:set1}, \eqref{eq:set2} 
are effectively semilinear due to 
\Cref{lem:0-+,lem:0star}, and hence the set
$\Phi_3 \ \circ \ \bigl(\stepp + \circ \ X_{q,r}\bigr)$,
and its intersection $Y_{q,r}$ with \eqref{eq:set2}, are
computable.
Knowing which remainders to expect,
one then finds the precise representation
$\minf{f(S)}$ of $f(S)$ by an \emph{exhaustive search}.
For every positively verified state-remainder pair $(q,r)$,
namely one with $Y_{q,r}\neq\emptyset$, 
one finds the least number $n$ such that
$q(n)\in f(S)$ and $n\equiv r \bmod D$.
For consecutive candidates $k=0, 1, \ldots$, 
apply the transformation \eqref{eq:back} to the set
\[
\bigl(\stepp + \circ \ \{q(r+k\cdot D)\}\bigr)
\]
and check nonemptiness of intersection with the set \eqref{eq:set2},
until a nonempty intersection is found.
This shows that the $D$-\saturated representation of $f(S)$ may be
computed, concluding the proof of the second part of \Cref{lem:compbackfor}.
\end{proof}

\section{Effective semilinearity of $\stepp{0}^*$: Proof of \Cref{lem:0star}} \label{sec:sl2}

% \begin{proof}[Proof of \Cref{lem:0star}]
    In outline,
    we prove that $\Phi_2 = {\stepp{0}^*}$ is essentially the reachability
    relation of an effectively constructible 1-\cm $\mach'$, to which we apply
    Parikh's theorem.

    \medskip

    The key observation which makes the construction of $\mach'$ possible,
    is that the relation $\stepp 0$ is \emph{essentially} one-dimensional:
    for any pair of zero-testing control states $q,q'\in Q$ and
    $c\in\setfromto{-C} C$, consider
    the set of pairs of initial and final counter values $(n,n')$ such that
    $n'=n+c$ and $q(n)\stepp 0 q'(n')$.
    We may describe this set in one dimension, defining
    \begin{align}\label{eq:per2}
     X^c_{q,q'} \ = \ \setof{n}{q(n)\stepp 0 q'(n+c)}.
    \end{align}
    % \mei{I flipped the sign in the above definition,
    % please check that $q'(n+c)$ is indeed what it should be
    % (rather than the $q'(n-c)$ that was there before).}
    Since the values of $n$ and $c$ together determine $n'$,
    the finite collection of sets $X^c_{q,q'}$ fully describes
    the $\stepp 0$ relation.

    Moreover, each set $X^c_{q,q'}$ is effectively semilinear,
    being the intersection of the effectively semilinear relation $\stepp 0$ 
    with the semilinear set $\Setof{\bigl(q(n),q'(n')\bigr)}{n'=n+c}$,
    projected down to one dimension.
    Hence, each $X^c_{q,q'}$ is
    ultimately periodic, with an effectively computable 
    period $N^c_{q,q'}$ and threshold $M^c_{q,q'}$.
    Define $N$ as the least common multiple of all the periods $N^c_{q,q'}$,
    across all choices of $c$, $q$ and $q'$.

    \para{Control states of $\mach'$}
    % We start by constructing a 1-\cm $\mach'$ for the non-iterated  
    % relation $\stepp 0$.
    Except for auxiliary states, to be clear from the construction below,
    the control states $Q'$ of $\mach'$ are the control states of $\mach$ 
    enriched with a number representing the remainder modulo $N$ of the 
    current counter value:
    \[
    Q' \ = \ Q \times \setfromto 0 {N-1} \ \cup \ \text{(auxiliary states)}.
    \]
    We write $q[r]$, instead of $(q,r)$, 
    for the control state of $\mach'$ corresponding 
    to $q\in Q$ and remainder $r<N$.
    As an invariant of our construction, 
    counter values $n$ reachable in state $q[r]$ in $\mach'$ always satisfy 
    $n\equiv r \bmod N$.
    % The initial control state of $\mach'$ is $q_0[0]$, where $q_0$ is the 
    % initial control state of $\mach$.

    The construction, % of transitions of the 1-\cm $\mach'$, 
    to be detailed below, ensures
    the following correctness claim, 
    for every pair $q,q'\in Q$ of zero-testing control states and
    $c\in\setfromto {-C} C$:
    \begin{align} \label{eq:mach}
    % \setof{(n,n')}
    {q(n) \stepp 0 {q}'(n+c)} \text{ in } \mach \ \iff \ 
    % \setof{(n,n')}
    {q[r](n)\auxstep^* q'[r'](n+c) \text{ in } \mach'},
    \end{align}
    where $r$ and $r'$ are the remainders of $n$ and $n+c$ modulo $N$,
    respectively, and $\auxstep^*$ denotes the subset of the reachability
    relation $\step^*$ of $\mach'$ where all the intermediate states are
    auxiliary.

    In other words, when $\mach'$ is restricted to runs whose intermediate
    states are all auxiliary,
    it realizes the non-iterated relation $\stepp 0$ of $\mach$.
    Once this restriction is lifted, $\mach'$ computes
    the iterated $\stepp{0}^\ast$ relation.
    
    \para{Transitions of $\mach'$}

    In control state $q[r]$, the 1-\cm $\mach'$ non-deterministically 
    chooses a number $c\in\setfromto{-C} C$, and a next zero-testing state 
    $q'\in Q$, such that
    the current counter value $n$ belongs to the set 
    $X^c_{q,q'}$ as defined in \eqref{eq:per2}.
    Upon positive test, $\mach'$
    increments the counter by $c$ and goes to the next state
    $q'[r + c \bmod N]$.

    As the remainder of $n$ modulo $N^c_{q,q'}$ is fixed and determined by 
    $r$,
    checking membership $n\in X^c_{q,q'}$ reduces to checking
    if $n\geq M^c_{q,q'}$
    (doable by decrementing the counter by $M^c_{q,q'}$, and then 
    incrementing it by $M^c_{q,q'}$ back)
    or checking if $n=m$ for some $m<M^c_{q,q'}$
   (doable similarly, by decrementing the counter by $m$, a zero-test, 
   and incrementing the counter by $m$ back).
   Auxiliary states are used as intermediate states to perform the above 
   steps.

  This construction assures the equivalence \eqref{eq:mach}. Considering
  now the unrestricted reachability relation $\step^*$ of $\mach'$,
  for every zero-testing control states $q,q'\in Q$ 
    % $c\in\setfromto {-C} C$, 
    and $n,n'\in\N$,
    %  we have:    
    \begin{align} \label{eq:mach2}
    {q(n) \stepp {0}^* {q}'(n')} \text{ in } \mach \ \iff \ 
    {q[r](n)\step^* q'[r'](n') \text{ in } \mach'},
    \end{align}
    where $r$ and $r'$ are the remainders of $n$ and $n'$ modulo $N$,
    respectively.
    The construction of $\mach'$ is effective.
    %  and its size is doubly exponential in $\mysize \mach$.
    
    \para{Effective semilinearity of $\stepp 0^*$}
    As 1-\cm{s} are a special case of pushdown automata,
    we may now obtain \Cref{lem:0star} by Parikh's Theorem:
    for each pair of zero-testing states $q, q' \in Q$,
    we build a pushdown automaton recognizing the language
    \[
        L_{q, q'} = \setof{ a^n b^{n'} }{ q(n) \stepp{0}^* q'(n') },
    \]
    or, equivalently,
    \[
        \setof{ a^n b^{n'} }{ q[n \bmod N](n) \step_{\mach'}^* q'[n' \bmod N](n') }.
    \]
    Applying Parikh's Theorem to this language, we'll get that
    the relation $\stepp{0}^*$ is effectively semilinear
    \cite{Parikh66,EGKL11}.

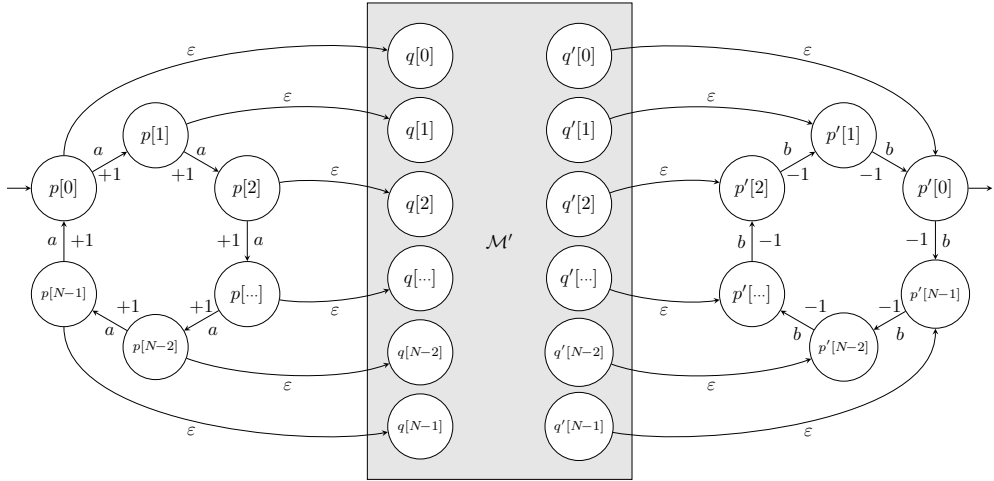
\begin{figure}[h]
\centering
\scalebox{0.7}
{\begin{tikzpicture}[
    >=stealth,
    state/.style={fill=white,state without output,minimum size=3.5em},
    initial text=,
    accepting/.style={accepting by arrow},
]
    \node[state,initial] (p0) at ($ (-360/6*0+150:2cm) $) {$p[0]$};
    \node[state]         (p1) at ($ (-360/6*1+150:2cm) $) {$p[1]$};
    \node[state]         (p2) at ($ (-360/6*2+150:2cm) $) {$p[2]$};
    \node[state]         (p3) at ($ (-360/6*3+150:2cm) $) {$p[{\scriptstyle\cdots}]$};
    \node[state]         (p4) at ($ (-360/6*4+150:2cm) $) {$\scriptstyle p[N\!-2]$};
    \node[state]         (p5) at ($ (-360/6*5+150:2cm) $) {$\scriptstyle p[N\!-1]$};

    \draw[fill=black!10] (4,4.5) rectangle (9,-4.5);
    \node at (6.5,0) {$\mach'$};

    \node[state] (q0) at (5,3.5)    {$q[0]$};
    \node[state] (q1) at (5,2.1)  {$q[1]$};
    \node[state] (q2) at (5,0.7)  {$q[2]$};
    \node[state] (q3) at (5,-0.7) {$q[{\scriptstyle\cdots}]$};
    \node[state] (q4) at (5,-2.1) {$\scriptstyle q[N\!-2]$};
    \node[state] (q5) at (5,-3.5)   {$\scriptstyle q[N\!-1]$};

    \draw[->] (p0) to[out=90,in=170,looseness=0.7]  node[above] {$\varepsilon$} (q0);
    \draw[->] (p1) to[out=20,in=160,looseness=0.7]  node[above] {$\varepsilon$} (q1);
    \draw[->] (p2) to[out=10,in=160,looseness=0.7]  node[above] {$\varepsilon$} (q2);
    \draw[->] (p3) to[out=-10,in=200,looseness=0.7] node[below] {$\varepsilon$} (q3);
    \draw[->] (p4) to[out=-20,in=200,looseness=0.7] node[below] {$\varepsilon$} (q4);
    \draw[->] (p5) to[out=270,in=190,looseness=0.7] node[below] {$\varepsilon$} (q5);

    \node[state] (qq0) at (8,3.5)  {$q'[0]$};
    \node[state] (qq1) at (8,2.1)  {$q'[1]$};
    \node[state] (qq2) at (8,0.7)  {$q'[2]$};
    \node[state] (qq3) at (8,-0.7) {$q'[{\scriptstyle\cdots}]$};
    \node[state] (qq4) at (8,-2.1) {$\scriptstyle q'[N\!-2]$};
    \node[state] (qq5) at (8,-3.5) {$\scriptstyle q'[N\!-1]$};

    \node[state,accepting] (pp0) at ($ (13,0) + (360/6*0+30:2cm) $) {$p'[0]$};
    \node[state]           (pp1) at ($ (13,0) + (360/6*1+30:2cm) $) {$p'[1]$};
    \node[state]           (pp2) at ($ (13,0) + (360/6*2+30:2cm) $) {$p'[2]$};
    \node[state]           (pp3) at ($ (13,0) + (360/6*3+30:2cm) $) {$p'[{\scriptstyle\cdots}]$};
    \node[state]           (pp4) at ($ (13,0) + (360/6*4+30:2cm) $) {$\scriptstyle p'[N\!-2]$};
    \node[state]           (pp5) at ($ (13,0) + (360/6*5+30:2cm) $) {$\scriptstyle p'[N\!-1]$};

    \draw[->] (qq0) to[out=10,in=90,looseness=0.7]  node[above] {$\varepsilon$} (pp0);
    \draw[->] (qq1) to[out=20,in=160,looseness=0.7]  node[above] {$\varepsilon$} (pp1);
    \draw[->] (qq2) to[out=20,in=170,looseness=0.7]  node[above] {$\varepsilon$} (pp2);
    \draw[->] (qq3) to[out=-20,in=190,looseness=0.7] node[below] {$\varepsilon$} (pp3);
    \draw[->] (qq4) to[out=-20,in=200,looseness=0.7] node[below] {$\varepsilon$} (pp4);
    \draw[->] (qq5) to[out=-10,in=270,looseness=0.7] node[below] {$\varepsilon$} (pp5);

    \draw[->] (p0) edge node[above left]  {$a$} node[below]         {$+1$} (p1);
    \draw[->] (p1) edge node[above]       {$a$} node[below left]    {$+1$} (p2);
    \draw[->] (p2) edge node[right]       {$a$} node[left]          {$+1$} (p3);
    \draw[->] (p3) edge node[below right] {$a$} node[above]         {$+1$} (p4);
    \draw[->] (p4) edge node[below]       {$a$} node[above right]   {$+1$} (p5);
    \draw[->] (p5) edge node[left]        {$a$} node[right]         {$+1$} (p0);

    \draw[->] (pp1) edge node[above]        {$b$} node[below left]   {$-1$} (pp0);
    \draw[->] (pp2) edge node[above left]   {$b$} node[below]        {$-1$} (pp1);
    \draw[->] (pp3) edge node[left]         {$b$} node[right]        {$-1$} (pp2);
    \draw[->] (pp4) edge node[below]        {$b$} node[above right]  {$-1$} (pp3);
    \draw[->] (pp5) edge node[below right]  {$b$} node[above]        {$-1$} (pp4);
    \draw[->] (pp0) edge node[right]        {$b$} node[left]         {$-1$} (pp5);
\end{tikzpicture}}
\caption{
    The construction we use to extract the language $L_{q,q'}$ out of
    the 1-\cm $\mach'$.
}
\label{fig:parikh}
\end{figure}
    It remains to describe the construction of the aforementioned pushdown automata,
    which we depict on \Cref{fig:parikh}.
    We extend the 1-\cm $\mach'$ to a \emph{labelled} 1-\cm, where
    every transition is labelled with a symbol from the alphabet $\{a,b\}$,
    or with $\varepsilon$.
    Fix a pair of zero-testing states $q,q'\in Q$ for which
    $L_{q, q'}$ is to be recognized.
    Take all the transitions of 
    $\mach'$ as unlabelled (labelled with $\varepsilon$), 
    add $N$ extra states $p[0] \ldots p[N-1]$ forming 
    an $a$-labelled cycle of transitions incrementing the counter by
    $1$:
    \[
    (p[i], \updinstr{1}, p[i+1  \bmod N]) \qquad \text{for all } 
    i\in\setfromto 0 {N-1}
    \] 
    (which assures storing the remainder modulo $N$ in states), 
    add another $N$ extra states $p'[N-1] \ldots p'[0]$ forming 
    a $b$-labelled cycle of transitions decrementing the counter by 1:
    \[
    (p'[i], \updinstr{-1}, p'[i-1 \bmod N]) \qquad \text{for all } 
    i\in\setfromto 0 {N-1},
    \]
    and finally add 0-update transitions from
    $p[i]$ to $q[i]$, and from $q'[i]$ to $p'[i]$:
    \[
    (p[i], \updinstr{0}, q[i]), 
    \quad 
    (q'[i], \updinstr{0}, p'[i]) \qquad \text{for all } 
    i\in\setfromto 0 {N-1}.
    \]
    A labelled 1-\cm extended in this manner recognizes
    the language $L_{q,q'}$ when considered
    with initial configuration $p[0](0)$ and final configuration $p'[0](0)$.
    By the previous discussion, the lemma follows by Parikh's theorem.

%% file: suffix.tex
\section{Target pumping: Proof of \Cref{lem:targetpump}}\label{sec:targetpump}
%

% We now prove \Cref{lem:targetpump}.
We focus on proving point $\punkt 2 2$, noticing however that
literally the same argument 
works for point $\punkt 1 2$.
% \sla{swapped $\punkt 1 2$ with $\punkt 2 2$}
The other two points, namely $\punkt 2 1$ and $\punkt 1 1$,
follow immediately by swapping the coordinates.
% \mei{I don't think the symmetry assumption is necessary
% for this particular step.}
% \sla{right, i've removed the symmetry assumption.}

To ease the presentation, we introduce the following definitions:
\begin{itemize}
\item For a transition $t = (q, \updinstr{w}, q') \in T$, the states $q, q'$
    are said to be the \emph{source} and \emph{target} of $t$ respectively,
    and the vector $w$ is said to be the \emph{effect} of $t$.
\item A \emph{path} is a sequence of transitions such that the target
    of each transition matches the source of the next transition
    in the sequence. Unlike a run, a path does not fully specify
    the source and target configurations.

    The source of the first transition is said to be the \emph{source} of the path,
    and the target of the final transition is said to be the \emph{target} of the
    path. The \emph{effect} of a path is defined as the sum of the effects of its
    transitions.
\item A \emph{cycle} is a path whose source and target coincide.
\item A path $\rho$ is said to be \emph{fireable} at a configuration $q(\vr v)$
    if there exists a configuration $q'(\vr v')$ such that $q(\vr v) \xstep{\rho} q'(\vr v')$.
    This amounts to requiring that $\rho$ can be executed starting from $q(\vr v)$,
    without the counters having to go negative along the way.
    We denote the target configuration $q'(\vr v')$ thus obtained by $q(\vr v); \rho$.
\item For paths $\rho_1$ and $\rho_2$, denote their concatenation by $\rho_1; \rho_2$.
% \item A path $\rho_2$ is said to be \emph{fireable after} a path $\rho_1$, if for any
%     configuration $q(\vr v)$ at which $\rho_1$ is fireable, the concatenation
%     $\rho_1; \rho_2$ is fireable as well.
\item For a path $\rho$, the \emph{transposed path} $\overline{\rho}$ is obtained
    by applying the symmetry assumption at each transition.
    If $\rho$ has effect $(a, b)$, $\overline{\rho}$ has effect $(b, a)$.
\end{itemize}

We define the \emph{norm} of a vector $\vr w = (w_1,w_2) \in \Z^2$ 
to be the maximal absolute value of its components,
$\norm{\vr w}=\max\{|w_1|,|w_2|\}$, and 
the norm of a 2-\vass $\V=(Q,T)$ to be the maximal norm of the effects 
of its transitions, 
$\norm \V = \max\{\norm{\vr w} \mid (q,\updinstr{\vr w},q')\in T\}$.

\begin{proof}[Proof of \Cref{lem:targetpump} $\punkt 2 2$]
Consider a fixed symmetric 2-\vass $\V = (Q,T)$ and
a run $\sigma : q(0,a)\step^* q'(0,a')$ for some $a, a'\in \N$.
Assuming that the difference $a'-a$ is sufficiently large
(namely larger than some $C$, depending only on $\V$),
we modify the run to lift the target configuration to $q'(0,a'+\gamma D_\sigma)$ for any
$\gamma \in \N$ and some particular $D_\sigma$, which depends on both $\V$ and $\sigma$.
Holding $\V$ fixed, $D_\sigma$ will only attain finitely many different values,
hence we argue that taking the least common multiple of all $D_\sigma$ yields a
constant $D$ satisfying the claim of the lemma.

Denote by $\widetilde \sigma$ the sequence of transitions fired along
the run $\sigma$, obtaining its corresponding path.
According to the linear-path-scheme (LPS) decomposition of 
\cite[Thm.~1]{Blondin15}, 
we may assume \mywlog that 
$\widetilde\sigma$ is built of paths 
$t_0, \ldots, t_\ell$ and cycles $c_1, \ldots, c_\ell$:
\begin{align}\label{eq:lps}
\widetilde\sigma \ \in \ t_0 \ c_1^* \ t_1 \ c_2^* \ \ldots \ c_\ell^* \ t_\ell,
\end{align}
where the lengths of each path $t_i$ and cycle $c_j$ are at most $e$,
and both $\ell$ and $e$ are bounded by a polynomial of the the size
$
 \mysize \V = 
\mysize Q + 
2 \cdot \norm \V \cdot \mysize {T}
% \, + \, \norm w \, + \, \norm{w'},
$
of the unary representation of $\V$.
Note that we do not claim any bound the number $n_j$ of iterations
in $\widetilde\sigma$ of the cycles $c_j$.

Let $C_0 \leq e \norm \V$ be the minimal value such that every cycle of $\V$
of length at most $e$ is \emph{fireable} at $\vr v_0=(C_0, C_0)$.
Call a cycle $c_j$ \emph{major} if its number of iterations $n_j$ is larger
than $2 C_0$, and \emph{minor} otherwise.
\Mywlog we may assume that all cycles are major,
by inlining any minor cycles $c_i$ into their neighboring one-off paths $t_{i-1}$ and $t_i$.
After this operation, the lengths of the one-off paths are bounded by
$e' := e (2C_0 \ell + \ell + 1)$.
% \sla{changed}
% \mei{nice catch, the previous value indeed wasn't enough.}
% $e' := e(\ell+1) + 2C_0\ell$. \sla{changed to 'sum of lengths'}

Call a cycle $c_j$ \emph{sum-increasing} if its effect $(x_j, y_j)$ satisfies $x_j + y_j > 0$.
We claim that, by setting a large enough value of $C$, we can guarantee that the run
will contain at least one (major) sum-increasing cycle.
Indeed, as we are assuming that $a' - a > C$, the sum of counters
must increase by more than $C$ along $\sigma$.
% In particular, the total effect of the one-off paths
% $t_0, \ldots, t_\ell$ is bounded,
% as each of them may increase the sum of the counters by at most $2 e \norm \V$.
% Likewise, a single iteration of a cycle may only increase the sum of the counters by at most
% $2 e \norm \V$, and each minor cycle gets executed at most $2 C_0$ times.
Thus, it is enough when
\[
    C \ \geq \ C_1 \ := \ 
    2\norm\V e'_\mathrm{max} (\ell_\mathrm{max} + 1),
\]
where $e'_\mathrm{max}$ and $\ell_\mathrm{max}$ are the 
(polynomial in $\mysize\V$) bounds
on $e'$ and $\ell$, respectively.
% \lukasz{Does it make sense to introduce $e'_\mathrm{max}$ and $\ell_\mathrm{max}$?
% It is enough to say that $\ell$ and $e'$ are polynomially bounded 
% to avoid introducing new variables.
% }
% \sla{let's keep them. at least $l_{max}$ we need}
% \mei{yup, as currently written, $e'$ and $\ell$ are run-specific, and
% then we take the max over all runs}

We now focus our attention on the sum-increasing cycles, and perform
case analysis in terms of the sign of each coordinate of their effect.
Specifically, the effect of a \emph{sum-increasing} cycle can be:
vertical $(0,+)$, horizontal $(+,0)$, or one of $(+,-)$, $(+,+)$, or $(-,+)$.
We classify these into three cases (\Cref{fig:cases}, left).

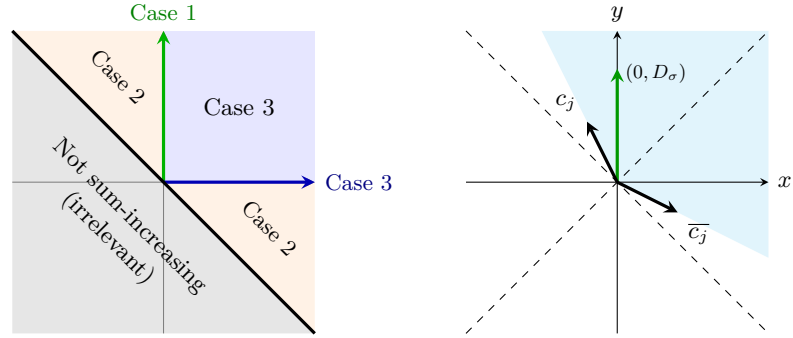
\begin{figure}[h]
\centering
\begin{tikzpicture}[
    >=stealth,
    every node/.style={font=\small}
]
\begin{scope}
    \fill[blue!10] (0,0) rectangle (2,2);
    \fill[orange!10] (0,0) -- (0,2) -- (-2,2) -- cycle;
    \fill[orange!10] (0,0) -- (2,0) -- (2,-2) -- cycle;
    \fill[black!10] (-2,-2) -- (2,-2) -- (-2,2) -- cycle;
    % axes
    \draw[->, black!50] (-2,0) -- (2,0);
    \draw[->, black!50] (0,-2) -- (0,2);

    \draw[->, very thick, green!70!black] (0,0) -- (0,2)
        node[above, green!50!black] {\footnotesize Case 1};
    \draw[->, very thick, blue!70!black] (0,0) -- (2,0)
        node[right, blue!50!black] {\footnotesize Case 3};
    \node at (1,1) {Case 3};
    \node[rotate=-45] at (1.4,-0.6) {\footnotesize Case 2};
    \node[rotate=-45] at (-0.6,1.4) {\footnotesize Case 2};

    % boundary of sum-increasing effects
    \draw[very thick] (-2,2) -- (2,-2);

    \node[rotate=-45,align=center] at (-0.6,-0.6)
            {Not sum-increasing \\ (irrelevant)};
\end{scope}
\begin{scope}[xshift=6cm]
    \fill[cyan!10] (0,0) -- (-1,2) -- (2,2) -- (2,-1) -- cycle;

    % axes
    \draw[->] (-2,0) -- (2,0) node[right] {$x$};
    \draw[->] (0,-2) -- (0,2) node[above] {$y$};

    % target direction
    \draw[->, very thick, green!60!black] (0,0) -- (0, 1.5)
        node[above right, black, pos=0.8, scale=0.8] {$(0, D_\sigma)$};

    % cycle effects
    \draw[->, very thick] (0,0) -- ($ 0.5*(-0.8, 1.6) $)
        node[above left] {$c_j$};
    \draw[->, very thick] (0,0) -- ($ 0.5*(1.6, -0.8) $)
        node[below right] {$\overline{c_j}$};

    % boundary of sum-increasing effects
    \draw[dashed] (-2,2) -- (2,-2);
    % axis of symmetry
    \draw[dashed] (-2,-2) -- (2,2);
        % node[above, pos=.9, rotate=45] {\scriptsize axis of}
        % node[below, pos=.9, rotate=45] {\scriptsize symmetry};
\end{scope}
\end{tikzpicture}
\caption{
    A view of the $(\Delta x, \Delta y)$ space of effects of individual cycles.
    \figleft An illustration of how the runs are split into
    three cases, based on the effects of the sum-increasing cycles
    present in the run.
    Cases 1--2 have \emph{priority}, in that as long as at least one
    cycle in the relevant area is present, the run is handled
    by one of these cases. On the other hand,
    Case 3 is the \emph{fallback}: it handles runs where
    \emph{all} sum-increasing cycles fall in the area labeled ``Case 3'' (blue).
    \figright Proof strategy for Case 2.
    Given a cycle $c_j$, we make use of the symmetry assumption
    to obtain the transposed cycle $\overline{c_j}$.
    The shaded area demarcates the effects we can obtain from linear
    combinations of $c_j$ and $\overline{c_j}$ with positive coefficients.
    Observe that the shaded area contains the desired vertical effect
    (green, not to scale).
}
\label{fig:cases}
\end{figure}

\para{Case 1: Some sum-increasing cycle $c_j$ has vertical $(0,+)$ effect}
One extra iteration of this cycle yields the claim, with the vertical effect
of the cycle determining $D_\sigma$.

\para{Case 2: Some sum-increasing cycle $c_j$ has effect $(-,+)$ or $(+,-)$}
Assume that $c_j$ has effect $(-x,y)$ for some positive integers $x,y$.
Observe that $y > x$.
The case of $(+,-)$ effect is analogous.

As an outline, we will seek to increase the iteration count $n_j$ of
the cycle $c_j$, and introduce into the run the transposed cycle
$\overline{c_j}$, such that the net effect of our modifications
becomes vertical (see \Cref{fig:cases}, right).

For our argument it is enough if the $n_j$ is large enough to guarantee that 
somewhere in the middle of $c_j^{n_j}$, the 
transposed cycle $\overline{c_j}$ becomes fireable.
We observe the following key claim, relying on the fact that  
the effect of $c_j$ is nonzero on both coordinates
(in particular note that we do not use the sum-increasing
assumption for this claim):
\begin{claim}  \label{claim:fireable}
    If $c_j^{n_j}$ is fireable at $q(\vr v)$, then for any
    repetition count $i$ satisfying
    $C_0 \leq i \leq n_j - C_0$,
    the transposed cycle $\overline{c_j}$ is fireable at $q(\vr v); c_j^i$.
\end{claim}
\begin{proof}
The $n_j+1$ configurations at the boundaries between
the iterations of $c_j$ within $c_j^{n_j}$
form an arithmetic progression given by
$(x_i, y_i) = (x_0 - i x, y_0 + i y)$.
%For a coordinate at which the
%effect of the cycle is positive, say $x > 0$, we obtain $x_i \geq C_0$
%for $i \geq C_0$. On the other hand, for a coordinate with a negative effect,
%say $y < 0$, we have $y_i = y_{n_j} - (n_j - i)y$, with $y_{n_j} \geq 0$, and
%thus $y_i \geq n_j - i$. Hence for $i \leq n_j - C_0$ we have $y_i \geq C_0$.
%As such, for $C_0 \leq i \leq n_j - C_0$, both coordinates are above $(C_0, C_0)$,
%which suffices to obtain that $\overline{c_j}$ is fireable.
%It suffices to show that for appropriate $i$ we have $x_i, y_i \geq C_0$.
%For coordinates at which the effect of $c_j$ is positive, say $x > 0$,
%It suffices that $C_0 \leq i \leq n_j - C_0$, as then we have
%$x_i = x_0 - i x \leq x_0 - i \geq C_0$ and 
Observe that for $C_0 \leq i \leq n_j - C_0$
we have
\begin{gather*}
    y_i \,=\, y_0 + iy \,\geq\, y_0 + i \,\geq\, C_0, \\
    x_i \,=\, x_0 - ix \,=\, x_{n_j} + (n_j - i) x
    \,\geq\, x_{n_j} + n_j - i \,\geq\, C_0.
    \qedhere
\end{gather*}
%
%Similarly, for coordinates with negative effect, say $y < 0$,
%it's enough that $i \leq n_j - C_0$.
\end{proof}

Recall now that, having eliminated all minor cycles by inlining,
we must have $n_j > 2C_0$. Thus,
the range of values of $i$ admissible for \Cref{claim:fireable}
is non-empty. Setting $i := C_0 + 1$, we have that the transposed cycle
$\overline{c_j}$ is fireable after the first $C_0 + 1$ iterations
of $c_j$ within $\sigma$. Equivalently, shifting the point of view
one iteration earlier, the concatenation $c_j; \overline{c_j}$
is fireable after $c_j^{C_0}$.

Recall we have that $y > x$, as the cycle is
sum-increasing. Thus, the effect of $c_j; \overline{c_j}$, being $(y-x, y-x)$,
is positive on both coordinates, and $c_j; \overline{c_j}$ may be fired arbitrarily
many times in the middle of $c_j^{n_j}$, without interference with the remainder of $\sigma$.
% Explicitly "in the middle", not "at the midpoint", the latter might not even exist
% if n_j happens to be odd.

% Let $k$ be the least common multiplicity of $m_2 - m_1$ and $m_1$.  
Now set $u:=x(y-x)$.
By iterating $x$ times the cycle $c_j; \overline{c_j}$,
we get a cycle with effect $(u,u)$, and then appending
$c_j^{y-x}$ we get the desired vertical effect.
Indeed, the effect of
\[
\big(c_j; \overline{c_j}\big)^{x}; c_j^{y-x}
\]
is $(0, (x+y)(y-x))$, and this cycle is still fireable in the middle of
$c_j^{n_j}$.
We thus get the claim with $D_\sigma = (x+y)(y-x)$.

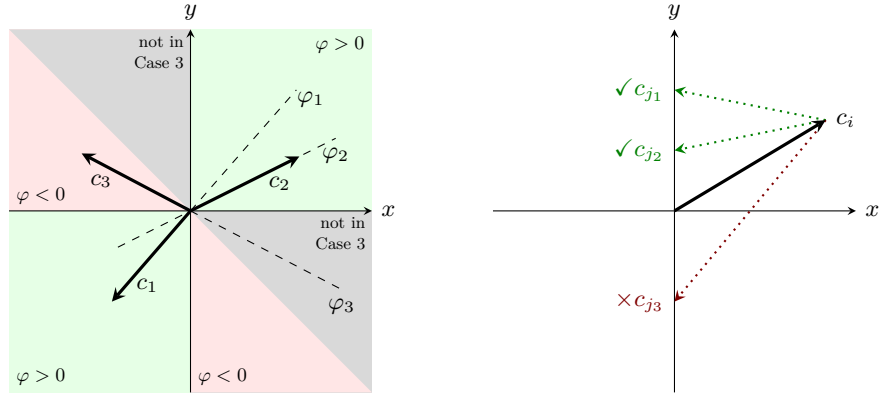
\begin{figure}[h]
\centering
\begin{tikzpicture}[
    >=stealth,
    scale=0.8,
    every node/.style={font=\small}
]
    \begin{scope}
        \fill[green!10] (0,0) rectangle (3,3);
        \fill[green!10] (0,0) rectangle (-3,-3);
        \fill[red!10] (0,0) rectangle (-3,3);
        \fill[black!15] (0,0) rectangle (3,-3);
        \fill[black!15] (0,0) -- (0,3) -- (-3,3) -- cycle;
        \fill[red!10] (0,0) -- (0,-3) -- (3,-3) -- cycle;
      % axes
      \draw[->] (-3,0) -- (3,0) node[right] {$x$};
      \draw[->] (0,-3) -- (0,3) node[above] {$y$};

      % rays corresponding to steepness
      \draw[dashed] (0,0) -- (1.75,2.0)
          node[pos=.92, right] {$\varphi_1$};

      \draw[dashed] (-1.2,-0.6) -- (2.4,1.2)
          node[pos=1.0, below] {$\varphi_2$};

      \draw[dashed] (0,0) -- (2.5,-1.3)
          node[pos=1.0, below] {$\varphi_3$};

      % cycle effects
      \draw[->, very thick] (0,0) -- (-1.3,-1.5)
          node[pos=.82, right] {$c_1$};

      \draw[->, very thick] (0,0) -- (1.8,0.9)
          node[pos=.82, below] {$c_2$};

      \draw[->, very thick] (0,0) -- (-1.8,0.95)
          node[pos=.82, below] {$c_3$};

      % angle / steepness indication
      % \node at (6.1,2.3) {$\varphi_1>\varphi_2>\varphi_3$};

      % labels
      % \node[below right] at (4.3,0.34) {horizontal effect};
      % \node[above left] at (-0.2,2.09) {vertical effect};

        % sign labels
        % \node[scale=0.8,below left] at (3,0) {$x_2>0$};
        % \node[scale=0.8,below right] at (-3,0) {$x_1,x_3<0$};

        \node[scale=0.8,above right] at (-3,-3) {$\varphi > 0$};
        \node[scale=0.8,below left] at (3,3) {$\varphi > 0$};
        \node[scale=0.8,above right] at (0,-3) {$\varphi < 0$};
        \node[scale=0.8,above right] at (-3,0) {$\varphi < 0$};
        \node[scale=0.7,below left,align=right] at (0,3) {not in \\ Case 3};
        \node[scale=0.7,below left,align=right] at (3,0) {not in \\ Case 3};
    \end{scope}

    \begin{scope}[xshift=8cm]
        \draw[->] (-3,0) -- (3,0) node[right] {$x$};
        \draw[->] (0,-3) -- (0,3) node[above] {$y$};

        \coordinate (a) at (2.5,1.5);
        \draw[->, thick, dotted, green!50!black] (a) -- (0, 2.0)
            node[left] {$\checkmark c_{j_1}$};
        \draw[->, thick, dotted, green!50!black] (a) -- (0, 1.0)
            node[left] {$\checkmark c_{j_2}$};
        \draw[->, thick, dotted, red!50!black] (a) -- (0, -1.5)
            node[left] {$\times c_{j_3}$};
        \draw[->, very thick] (0,0) -- (a)
            node[right] {$c_i$};
    \end{scope}
\end{tikzpicture}
\caption{
  The steepness comparison central to the proof strategy of Case 3.
  \figleft An illustration of the properties of steepness.
  The vectors represent the effects of three cycles
  $c_1$, $c_2$, and $c_3$, where $c_2$ is horizontally
  positive, and the other two are horizontally negative.
  The cycle $c_1$ is steeper than $c_2$, and $c_2$ is steeper than $c_3$.
  Cycles whose effect lies within the area shaded green (resp. red) have positive (resp. negative) steepness.
  The area shaded gray corresponds to sum-increasing cycles that get handled in Case 2,
  and thus cannot occur in Case 3.
  \figright A sum-increasing cycle $c_i$,
  with three candidates for a subsequent horizontally-negative cycle $c_j$ (cf. \Cref{claim:steep}),
  illustrated in terms of their combined effect.
  %(when combined with coefficients chosen to put said effect on the vertical axis).
  Observe how the vertical effect obtained this way is positive exactly
  when $c_i$ is steeper than $c_j$.
}
\label{fig:steep}
\end{figure}

\para{Case 3: All sum-increasing cycles have $(+,0)$ or $(+,+)$ effect}
Let us restrict our attention to those cycles $c_j$, whose effect 
$(x_j,y_j)$ satisfies $x_j\neq 0$.
Depending on whether $x_j>0$ or $x_j<0$, we call such cycles 
\emph{horizontally positive}, or \emph{horizontally negative} respectively.
Define the \emph{steepness} of such cycles $c_j$ as $y_j/x_j$, and denote it by $\varphi_j$.
Note that, while in general the steepness may be negative,
by assumption it will never be negative for a sum-increasing cycle
(see \Cref{fig:steep}, left).

We will now aim at increasing 
the repetition count of one of the sum-increasing
cycles, together with one other cycle, occurring later in the run, 
in order to obtain a vertical effect (see \Cref{fig:steep}, right).
That suitable cycles exist is stated by the following claim.

\begin{figure}[t]
\centering
\begin{tikzpicture}[
    >=stealth,
    scale=1.0,
    every node/.style={font=\small}
]

% ------------------------------------------------------------
% Run in configuration space
% ------------------------------------------------------------

% Axes
\draw[->] (0,0) -- (11.8,0) node[right] {$x$};
\draw[->] (0,0) -- (0,4.5) node[above] {$y$};

% Configurations along the run
\coordinate (P0) at (0.0,1.5);
\coordinate (P1) at ($ (P0) + (0.5,-0.6) $);
\coordinate (P1e) at ($ (P1) + (1.7,0.2) $);
\coordinate (P2) at ($ (P1e) + (0.5,0.8) $);
\coordinate (P2e) at ($ (P2) + (2.3,2.3) $);
\coordinate (P3) at ($ (P2e) + (1.0,0.5) $);
\coordinate (P3e) at ($ (P3) + (-0.8,-3) $);
\coordinate (P4) at ($ (P3e) + (1.0,0.2) $);
\coordinate (P4e) at ($ (P4) + (1.5,1.0) $);
\coordinate (P5) at ($ (P4e) + (0.5,1.0) $);
\coordinate (target) at (0,4);

% Ordinary path segments
\draw[dotted,->] (P0) -- (P1);
\draw[->,very thick,green!70!black] (P1) -- (P1e) node[pos=.55,above,black] {$c_1^{n_1}$};
\draw[dotted,->] (P1e) -- (P2);
\draw[->,very thick,green!70!black] (P2) -- (P2e) node[pos=.45,right,black] {$c_2^{n_2}$};
\draw[dotted,->] (P2e) -- (P3);
\draw[->,red!70!black] (P3) -- (P3e) node[pos=.55,right,black] {$c_3^{n_3}$};
\draw[dotted,->] (P3e) -- (P4);
\draw[->,green!70!black] (P4) -- (P4e) node[pos=.45,above,black] {$c_4^{n_4}$};
\draw[dotted,->] (P4e) -- (P5);

% ------------------------------------------------------------
% Configuration labels
% ------------------------------------------------------------

\fill (P0) circle (1.5pt);
\fill (P1) circle (1.5pt);
\fill (P2) circle (1.5pt);
\fill (P3) circle (1.5pt);
\fill (P1e) circle (1.5pt);
\fill (P2e) circle (1.5pt);
\fill (P3e) circle (1.5pt);
\fill (P4) circle (1.5pt);
\fill (P4e) circle (1.5pt);
\fill (P5) circle (1.5pt);
\fill (target) circle (1.5pt);

\node[scale=0.7,left] at (P0) {$(0, a)$};
\node[scale=0.7,left] at (target) {$(0, a')$};
\node[scale=0.7,below right]      at (P1) {$(X^\pre_1, Y^\pre_1)$};
\node[scale=0.7,below right]      at (P1e) {$(X^\post_1, Y^\post_1)$};
\node[scale=0.7,below right]      at (P2) {$(X^\pre_2, Y^\pre_2)$};
\node[scale=0.7,above left]      at (P2e) {$(X^\post_2, Y^\post_2)$};
\node[scale=0.7,right]      at (P3) {$(X^\pre_3, Y^\pre_3)$};
\node[scale=0.7,below] at (P3e) {$(X^\post_3, Y^\post_3)$};
\node[scale=0.7,below right]      at (P4) {$(X^\pre_4, Y^\pre_4)$};
\node[scale=0.7,below right] at (P4e) {$(X^\post_4, Y^\post_4)$};
\node[scale=0.7,right]      at (P5) {$(X^\pre_5, Y^\pre_5)$};

  % angle / steepness indication
  % \node at (11,5.0) {$\varphi_1<\varphi_2<\varphi_3$};

% invariant lines
\draw[dashed,thick,blue!70!black] ($ (0,2) + -0.2*(1.7,0.2) $) -- +($ 1.7*(1.7,0.2) $);
\draw[dashed,thick,blue!70!black] ($ (0,2.5) + -0.2*(1,1) $) -- +($ 2.7*(1,1) $);
\end{tikzpicture}

\caption{A run prefix that does not yet satisfy \Cref{claim:steep}.
Solid lines represent iterated cycles, and dotted lines represent 
one-off paths. Record-breaking cycles are drawn using thicker lines;
the color of the cycle indicates whether the cycle increases
or decreases the counter sum.
Parallel to each record-breaking cycle, dashed lines represent
the successive invariant lines, below which the configuration
will stay.
Note the target configuration $(0,a')$, above both invariant lines.
% The cycles $c_1$ and $c_2$ are 
% horizontally positive and sum-increasing, the latter steeper than
% the former, and $c_3$ is a horizontally negative (sum-decreasing) cycle
% steeper than both $c_1$ and $c_2$.
% Thus none of sum-increasing cycles is followed by a less steep
% horizontally negative cycle.
% Observe how this run gets further away from the Y axis as it progresses.
% Intuitively, runs violating \Cref{claim:steep} like this one are ruled out, because
% we're considering runs that end \emph{on} the Y axis,
% above the starting point of the run.
% To extend the run and have it reach such an endpoint, one would need to extend it
% with a horizontally-negative cycle that's either 
% less steep than $c_2$, or sum-increasing, with the latter option
% putting the run in the purview of Case 2.
}
\label{fig:steep-counterexample}
\end{figure}
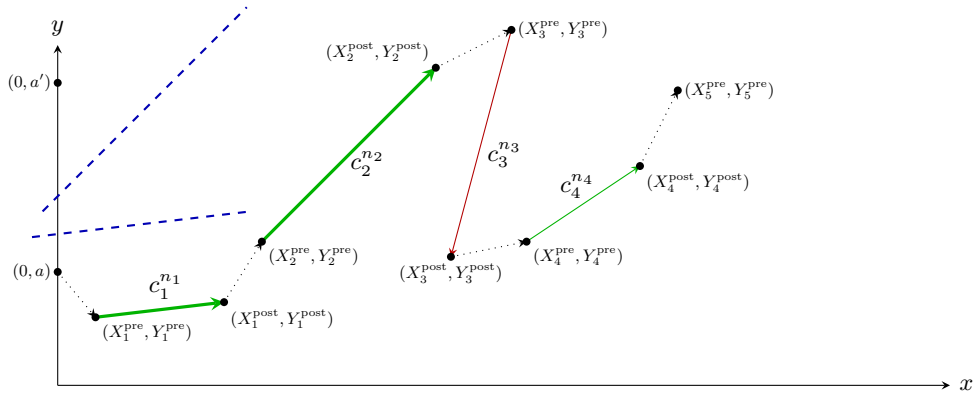

\begin{claim} \label{claim:steep}
    There are some $i<j$ such that $c_i$ is a sum-increasing cycle, $c_j$ is
    horizontally negative, and $c_i$ is steeper than $c_j$ (\ie $\varphi_i > \varphi_j$).
\end{claim}
\begin{proof}
%Then, for any horizontally negative cycle $c_j$,
Call a sum-increasing cycle \emph{record-breaking} if it is
steeper than all sum-increasing cycles that came before it,
and suppose that the claim doesn't hold.
In outline, we will define a sequence of straight lines in the configuration space,
each parallel to the effect of the most recent record-breaking cycle,
such that the counter values will always stay below the current line
(see \Cref{fig:steep-counterexample}).
We will then show that the supposed target configuration of the run
is \emph{above} all the invariant lines, to derive a contradiction.

For a cycle $c_i$, denote by $\beta_i$ the steepness of
the most recent record-breaking cycle, \ie
\[
    \beta_i = \begin{cases}
        \max(\beta_{i-1}, \varphi_i) & \text{if $c_i$ sum-increasing} \\
        \beta_{i-1} & \text{otherwise}
    \end{cases}
\]
with $\beta_0 = 0$ as the initial value.
As one would expect, the value of $\beta_i$ will define the slope of each line.
Concretely specifying the $y$-intercept, however, requires a bit more work.

Define $(X^\pre_i, Y^\pre_i)$ and $(X^\post_i, Y^\post_i)$ to be the counter values
the run achieves immediately before and after $c_i^{n_i}$, respectively.
Let $k$ be the index of the first sum-increasing cycle of the run.
For $i \geq k$, define $S_i \leq C_1$ as the sum of the norms of
the effects of all one-off paths between $c_k$ and $c_i$, \ie
writing $\norm {t}$ for the sum of norms of transitions along a path $t$,
\[
    S_i = \sum_{j=k}^{i-1} \norm{t_j}
\]
We extend the definitions in a natural way to $i = \ell + 1$ where relevant:
let $(X^\pre_{\ell + 1}, Y^\pre_{\ell + 1})$ be the counter values
of the run's target configuration, and denote by $S_{\ell + 1}$
the sum of norms of effects of all one-off paths after $c_k$, up to and
including $t_\ell$.

% Note that, if \Cref{claim:steep} indeed doesn't hold, we have $\beta_j \leq \varphi_j$ for any
% horizontally negative cycle $c_j$.
% Indeed, for $c_j$ occurring after $c_{k_1}$, we would otherwise get
% $\varphi_i = \beta_j > \varphi_j$, which satisfies \Cref{claim:steep}.
% On the other hand, for $c_j$ before the first sum-increasing cycle,
% we only need to show $-1 \leq \varphi_j$. But horizontally negative
% cycles of steepness lower than $-1$ are sum-increasing, and for
% the current case we are assuming that all sum-increasing cycles
% are horizontally positive.
%
% On the other hand, for a horizontally \emph{positive} cycle $c_j$,
% we have $\varphi_j \leq \beta_j$, either by definition, for sum-increasing cycles,
% or trivially, as $\varphi_j \leq -1 \leq \beta_j$ for 
% horizontally positive cycles $c_j$ that are not sum-increasing.

By induction on $i \geq k$, we will now prove that the pair $(X^\pre_i,Y^\pre_i)$ of
counter values will always satisfy
\begin{equation} \label{eq:steep-ind}
    Y^\pre_i \, \leq \, Y^\pre_k \, + \, S_i (1 + \beta_i) \, + \, \beta_i X^\pre_i.
\end{equation}
Geometrically, this defines the dashed lines in \Cref{fig:steep-counterexample}.

At $i = k$, the inequality trivially holds.
Suppose then that the inequality holds at a specific $i \geq k$,
aiming to prove it for $i+1$.

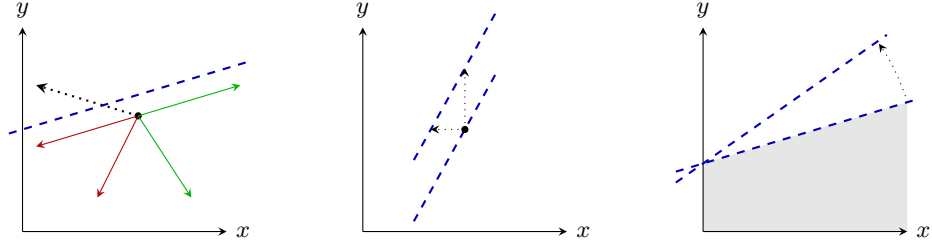
\begin{figure}[t]
\centering
\begin{tikzpicture}[
    >=stealth,
    scale=0.9,
    every node/.style={font=\small}
]
    \begin{scope}
        \draw[->] (0,0) -- (3,0) node[right] {$x$};
        \draw[->] (0,0) -- (0,3) node[above] {$y$};

        \draw[dashed,thick,blue!70!black] ($ (0,1.5) + -0.2*(1,0.3) $) -- +($ 3.5*(1,0.3) $);

        \coordinate (conf) at (1.7,1.7);
        \fill (conf) circle (1.5pt);
        \draw[->, green!70!black] (conf) -- +($ 1.5*(1,0.3) $);
        \draw[->, green!70!black] (conf) -- +($ 0.6*(1.3,-2) $);
        \draw[->, red!70!black] (conf) -- +($ -1.5*(1,0.3) $);
        \draw[->, red!70!black] (conf) -- +($ 0.6*(-1,-2) $);
        \draw[->, dotted, thick] (conf) -- +($ 1.5*(-1,0.3) $);
    \end{scope}

    \begin{scope}[xshift=5cm]
        \draw[->] (0,0) -- (3,0) node[right] {$x$};
        \draw[->] (0,0) -- (0,3) node[above] {$y$};

        \coordinate (conf) at (1.5,1.5);
        \coordinate (slope) at (0.5,0.9);
        \coordinate (shift) at (-0.5,0);
        \draw[dashed,thick,blue!70!black] ($ (conf) + -1.5*(slope) $) -- +($ 2.5*(slope) $);
        \draw[dashed,thick,blue!70!black] ($ (conf) + (shift) + -0.5*(slope) $) -- +($ 2.5*(slope) $);
        \draw[->, dotted] (conf) -- +(shift);
        \draw[->, dotted] (conf) -- +(0,0.9);
        \fill (conf) circle (1.5pt);
    \end{scope}

    \begin{scope}[xshift=10cm]
        \coordinate (root) at (0,1);
        \coordinate (conf) at (1.5,0.8);
        \coordinate (slope1) at (1,0.3);
        \coordinate (slope2) at (1,0.7);

        \fill[black!10] (0,0) -- (root) -- ($ (root) + 3*(slope1) $) -- (3,0) -- cycle;

        \draw[->] (0,0) -- (3,0) node[right] {$x$};
        \draw[->] (0,0) -- (0,3) node[above] {$y$};

        \draw[dashed,thick,blue!70!black] ($ (root) + -0.4*(slope1) $) -- +($ 3.5*(slope1) $);
        \draw[dashed,thick,blue!70!black] ($ (root) + -0.4*(slope2) $) -- +($ 3.1*(slope2) $);
        \draw[->,dotted] ($ (root) + 3*(slope1) $) arc (16.7:35:3cm);
        % \fill (conf) circle (1.5pt);
    \end{scope}
\end{tikzpicture}
\caption{
    A geometric intuition for why the inductive invariant gets preserved.
    \figleft Horizontally positive cycles (green) stay below the invariant line,
    as by construction they are at most as steep as the invariant's slope;
    horizontally negative cycles either stay below the invariant line (red)
    or satisfy \Cref{claim:steep} (dotted black).
    \figcenter One-off paths may force the invariant line to be raised,
    but this can be bounded ahead of time. 
    % The greatest upwards
    % movement per unit effect is forced by a path %whose effect is pointing
    % going directly left, when the slope of the invariant line is high.
    \figright Because the configuration is guaranteed to stay at $x \geq 0$,
    we may freely make the invariant line steeper as long as the $y$-intercept
    doesn't get lowered in the process.
}
\label{fig:induction}
\end{figure}
We shall follow the structure of the run, aiming first to prove
that after executing the effect of $c_i$, the configuration
$(X^\post_i, Y^\post_i)$ stays below the same line, \ie
\begin{equation} \label{eq:steep-ind-post}
    Y^\post_i \, \leq \, Y^\pre_k \, + \, S_i (1 + \beta_i) \, + \, \beta_i X^\post_i.
\end{equation}
See \Cref{fig:induction} (left) for an illustration of this step.
Symbolically, writing $X^\post_i = X^\pre_i + n_i x_i$ and $Y^\post_i = Y^\pre_i + n_i y_i$, we see that
our aim is to prove $y_i \leq \beta_i x_i$.
The first case, easily handled, is when $x_i = 0$. Indeed, we handle cycles of vertical effect
$(0, +)$ in Case 1, so if we encounter such a cycle here, then $y_i \leq 0$ must hold.

In the interesting case where $x_i \neq 0$, we apply the definition
of $\varphi_i$ to express the desired inequality as $\varphi_i x_i \leq \beta_i x_i$,
and consider the sign of $x_i$.
For a horizontally positive cycle, \ie $x_i > 0$, we need $\varphi_i \leq \beta_i$.
If $c_i$ is sum-increasing, this holds by definition;
otherwise note $\varphi_i \leq 0 \leq \beta_i$.

Finally, let $c_i$ be a horizontally negative cycle. In that situation, we need
$\varphi_i \geq \beta_i$.
Let $j$ be the most recent record-breaking cycle before $i$, which must
exist, as $i \geq k$ (and $i = k$ is ruled out by $c_i$ being horizontally negative).
Now note that if we had $\varphi_i < \beta_i = \varphi_j$, cycles $c_j$ and $c_i$
would satisfy \Cref{claim:steep}. As we are assuming that the claim doesn't hold,
we indeed have $\varphi_i \geq \beta_i$, thereby establishing \eqref{eq:steep-ind-post}.
% \mei{$i$ and $j$ ended up having the opposite meaning to the statement of the claim,
% maybe we should swap them.}
% \sla{no need to swap}

Our next focus shall be handling the effect of $t_i$ (\Cref{fig:induction}, center), proving that
\begin{equation} \label{eq:steep-ind-pre1}
    Y^\pre_{i+1} \, \leq \, Y^\pre_k \, + \, S_{i+1} (1 + \beta_i) \, + \, \beta_i X^\pre_{i+1}.
\end{equation}
With $(\Delta x, \Delta y)$ denoting the effect of the one-off path $t_i$,
we have $X^\pre_{i+1} = X^\post_i + \Delta x$, and analogously,
$Y^\pre_{i+1} = Y^\post_i + \Delta y$. Comparing with \eqref{eq:steep-ind-post}, we
see that it'll suffice to show
\[
    \Delta y \, \leq \, \norm{t_i} (1 + \beta_i) \, + \, \beta_i \Delta x
\]
(recall that $S_{i+1} = S_i + \norm{t_i}$). As both $|\Delta x|$ and $|\Delta y|$
are bounded by $\norm{t_i}$, we have
\begin{align*}
    \Delta y \, &\leq \, \norm{t_i}, \\
    - \Delta x \, &\leq \, \norm{t_i}.
\end{align*}
Multiplying the latter by $\beta_i$ and adding the two inequalities together,
we obtain
\[
    \Delta y - \beta_i \Delta x \, \leq \norm{t_i} (1 + \beta_i),
\]
as required.

Having now established \eqref{eq:steep-ind-pre1}, it
remains to raise the slope of the invariant line to $\beta_{i+1}$
(compare \Cref{fig:induction}, right).
We note that
$\beta_i \leq \beta_{i+1}$, and moreover both
$S_{i+1}$ and $X^\pre_{i+1}$ are nonnegative, so we have
\[
    \begin{split}
        Y^\pre_{i+1} \, &\leq \, Y^\pre_k \, + \, S_{i+1} (1 + \beta_i) \, + \, \beta_i X^\pre_{i+1}
    \\
        &\leq \, Y^\pre_k \, + \, S_{i+1} (1 + \beta_{i+1}) \, + \, \beta_{i+1} X^\pre_{i+1}.
    \end{split}
\]
This ends the induction, and so we get \eqref{eq:steep-ind} up to and including $i = \ell + 1$.
As $(X^\pre_{\ell + 1}, Y^\pre_{\ell + 1}) = (0, a')$ is the final configuration, we have:
\[
    a' \, \leq \, Y^\pre_k \, + \, S_{\ell + 1} (1 + \beta_{\ell + 1}).
\]

Let us now focus on establishing an upper bound for each expression occurring
on the right-hand side. Firstly, we have $S_{\ell + 1} \leq C_1$, which should
be clear from the definition of the latter.
Next, note that for any cycle $c_i$ we have $\varphi_i = y_i / x_i \leq |y_i| \leq e' \norm\V$,
and therefore $\beta_{\ell + 1} \leq e' \norm\V$.

Finally, for $Y^\pre_k$, we note that $Y^\pre_k \leq X^\pre_k + Y^\pre_k$, and analyze
the latter instead. Recall that $(X^\pre_k, Y^\pre_k)$ are the counter values immediately
before the first sum-increasing cycle of the run, so any increase above the initial
counter sum of $a$ is fully attributable to the effects of the one-off paths before $c_k$.
Thus, we get $Y^\pre_k \leq a + C_1$. In summary, we obtain
\[
    a' \, \leq \, a \, + \, C_1 \, + \, C_1 (1 + e' \norm \V) \,=\, a \, + \, C_1 (2 + e' \norm \V).
\]
As the assumption of the lemma is that $a' - a > C$, we can obtain a contradiction
by setting
\[
    C := C_1 (2 + e'_\mathrm{max} \norm \V).
    \qedhere
\]
\end{proof}

Take any cycles $c_i$, $c_j$, where $i<j$, as given by \Cref{claim:steep}. 
Their horizontal effects
satisfy $x_i>0$ and $x_j<0$, and hence
%    positive integers $I, J$, 
the sum of the effects of $c_i^{-x_j}$ and $c_j^{x_i}$ is vertical and equal 
$(0,h)$ 
for $h = -x_j y_i + x_i y_j > 0$.
Expanding $c_i^{n_i}$ with $c_i^{-x_j}$ and $c_j^{n_j}$ with $c_j^{x_i}$ 
yields a run ending in $q'(0,a'+D_\sigma)$, where $D_\sigma = h$.
Moreover, this modification will never decrease any intermediate value
of the counters, as compared to the original run $\sigma$.
The proof of Case 3 is thus completed.
%
% The constant $C$ must be large enough so that, even after forgetting of
% the total effects of all transitions $t_0, \ldots, t_\ell$, the iterations of
% cycles $c_1, \ldots, c_\ell$ increase the sum by the amount large enough so that
% one of sum-increasing cycles is iterated at least $2 C_0$ times (are required
% by \Cref{claim:fireable}).
% As $C_0$ is at most $e\cdot \norm \V$ and $\ell$ is polynomial in $\mysize \V$,
% a polynomial value of $C$, \wrt $\mysize \V$, 
% is sufficient to satisfy the above requirement.
% \sla{complete the calculations}

\medskip

To summarize, in all the three cases 
the sufficiently large value of $C$ is $C := C_1 (2 + e'_\mathrm{max}\norm\V)$.
The proof of \Cref{lem:targetpump} is thus completed.
\end{proof}

\section{Undecidability for $\cycg 3$-\cm{s}} \label{sec:undecid}

We now turn to $\cycg 3$-counter machines, and prove that their reachability
problem is undecidable, in contrast with symmetric 2-\cm{s}.
We proceed by reducing from the reachability problem of general 2-\cm{s},
which is undecidable \cite{Minsky67}.
Our reduction readily generalizes to a reduction from general $(n-1)$-\cm{s}
to $\cycg{n}$-\cm{s}, showing that reachability is undecidable
in $\cycg{n}$-\cm{s}, 
% have undecidable reachability problem 
for $n \geq 3$. For concreteness, however, we focus
on the case of $n = 3$.

\begin{proof}[Proof of \Cref{thm:undecid}]
Recall that a $\cycg 3$-\cm is a 3-\cm such that
%  satisfying the condition that,
for any transition labeled by $\updinstr{x,y,z}$, the transitions
$\updinstr{y,z,x}$ and $\updinstr{z,x,y}$ are also present
between the same pair of states. Note however, that the $\cycg 3$-\cm
we construct here also satisfies a stronger condition:
the $\xztinstr$ transitions are symmetric as well, effectively testing
that an \emph{unspecified} counter is zero. We write this as
$\ztinstr{\ast}$.

For any 2-\cm $\mach = (Q, T)$, we shall construct a $\cycg 3$-\cm
$\mach' = (Q', T')$ that simulates it,
where $Q'$ is a copy of $Q$ with auxiliary states adjoined.

For a configuration $q(a, b)$ of $\mach$, the corresponding configuration
of $\mach'$ shall have the form $q(3a+2, 3b+2, 0)$,
with two \emph{main counters} keeping track of the values
of the counters of $\mach$,
and one \emph{auxiliary counter}, which is always zero outside of
auxiliary states.
Thanks to the presence of the auxiliary counter,
we will be able to ``anchor'' each transition as it executes,
in such a way that any run making use of an unwanted cyclic rotation
of a transition
will get rapidly cut short by a failing zero-test.

To simulate a transition of the form $(q, \updinstr{x, y}, q') \in T$,
we include in $T'$ the transitions in use by the following fragment of a run:
\[ \begin{split}
    q&(3a + 2, && 3b + 2, && 0)
    \xstep{\updinstr{3x - 1, 3y - 1, 0}} \\
    &(3(a + x) + 1, && 3(b + y) + 1, && 0)
    \xstep{\ztinstr{\ast}} \\
    &(3(a + x) + 1, && 3(b + y) + 1, && 0)
    \xstep{\updinstr{1, 1, 0}} \\
    &(3(a + x) + 2, && 3(b + y) + 2, && 0)
    \xstep{\ztinstr{\ast}} \\
    q'&(3(a + x) + 2, && 3(b + y) + 2, && 0),
\end{split} \]
as well as their cyclic rotations.
We now need to verify that the cyclic shifts of these transitions
do not cause any issues.

Let us first consider the $\updinstr{3x-1, 3y-1, 0}$ instruction.
As the update vector contains only one zero, executing any cyclic rotation will
disturb the auxiliary counter, making it non-zero.
Moreover, because the main counters have a remainder of $2 \bmod 3$
before the transition, they can only have a remainder of $1$ or $2$ after
any cyclic rotation executes. In particular, none of the main counters may
become zero.
Hence, should any unwanted cyclic rotation of the first update instruction execute,
the zero-test that follows will fail.

The reasoning for the $\updinstr{1, 1, 0}$ instruction is similar:
any cyclic rotation will increase the auxiliary counter to 1,
and neither of the main counters may decrease to 0,
as the vector is non-negative on all coordinates.

We now turn to simulating the zero-test transitions of $\mach$.
For concreteness, let us focus on the case where the first counter
is being zero-tested;
a zero-test of the second counter is handled analogously.

To simulate a transition $(q, \ztinstr{1}, q') \in T$ of $\mach$,
$\mach'$ performs the following steps:
\[ \begin{split}
    q&(3a + 2, && 3b + 2, && 0)
    \xstep{\updinstr{-2, -1, +2}} \\
    &(3a, && 3b + 1, && 2)
    \xstep{\ztinstr{\ast}} \\
    &(0, && 3b + 1, && 2)
    \xstep{\updinstr{+2, +1, -2}} \\
    &(2, && 3b + 2, && 0)
    \xstep{\ztinstr{\ast}} \\
    q'&(2, && 3b + 2, && 0).
\end{split} \]
As the instruction $\updinstr{-2, -1, +2}$ has only one non-negative number,
any of its other rotations will attempt to make the auxiliary counter negative,
and thus will not be able to execute. Then, the first $\ztinstr{\ast}$
is performed. As the counters have remainders $(0, 1, 2) \bmod 3$ respectively,
this zero-test will only succeed when $3a = 0$, exactly as required
to simulate a $\ztinstr{1}$ instruction.

To finish simulating the zero-test, $\updinstr{+2, +1, -2}$ undoes
the modifications of the counters. As we now know that the first counter is
zero, and all numbers in the update instruction are non-zero, we can infer that
no rotation can leave the first counter at zero. Moreover, the second counter
has a remainder $1 \bmod 3$, therefore the sole negative number in the update
instruction, having value $-2$, cannot
make that counter be zero either. Thus, for the final zero-test to pass,
it is the auxiliary counter that must become zero, and this can only happen
if the $-2$ is appropriately aligned, which rules out all unwanted rotations.
%  other than the intended one.
\end{proof}

%% file: variants.tex
\section{Variants} \label{sec:variants}

We now discuss the consequences of imposing the following 
restrictions on symmetric \cm{s}:
\begin{itemize}
\item the set of $\xztinstr$ transitions is also \emph{symmetric}, 
that is invariant under permutations of counters; 
\item all $\xupdinstr$ transitions $(q,\updinstr{\vr w},q')$ are 
\emph{single-counter}, that is $\vr w$  is nonzero on at most one coordinate.
\end{itemize}
\para{Symmetric zero-tests}

For simplicity, let's focus on 2-\cm{s} first.
Symmetric $\xztinstr$ transitions, effectively testing
that an \emph{unspecified} counter is zero,
written $\ztinstr{\ast}$,
can surprisingly simulate unrestricted ones.
To this aim, a configuration $q(v_1,v_2)$ of a symmetric 2-\cm would 
be simulated by $q(2v_1,2v_2+1)$, thus doubling the counter values and 
distinguishing them by parity.
Consequently, a transition $(q,\updinstr{w_1,w_2},q')$ is simulated
a transition $(q,\updinstr{2w_1,2w_2},q')$.
Therefore the $\xupdinstr$ transitions, even symmetric ones, preserve
parity of both counters.
Due to parity distinction, a transition
$(q,\ztinstr{\ast},q')$ checking if an unspecified counter is 0,
in fact checks if counter 1 is 0.
On the other hand,
simulation of transition $(q,\ztinstr 2, q')$
is done in three steps,
using two auxiliary intermediate states $p_t,p'_t$:
\begin{align} \label{eq:sim3}
(q,\updinstr{1,-1},p_t) \qquad
(p_t, \ztinstr{\ast}, p'_t) \qquad
(p'_t, \updinstr{-1,1},q').
\end{align}
Observe that the first and the last transition swap parity of
counters, and likewise do their swapped versions.
Therefore a successful $\ztinstr{\ast}$ is only possible on
counter 2, which rules out possibility of use of swapped versions
of the first and last transition.

\smallskip

The above considerations generalize 
to $G$-\cm{s}, for 
an arbitrary number $d>2$ of counters, and arbitrary permutation group
$G\leq\symg d$.
To this aim we distinguish between the counters by remainders modulo 
$d$, simulate a configuration $q(a_1, \ldots, a_d)$ by 
\[q(d a_1, d a_2 + 1, \ldots, d a_d + d-1),\]
and simulate a transition $(q,\updinstr{\vr v}, q')$ by
$(q,\updinstr{d\vr v}, q')$.
Furthermore, to correctly simulate a zero-test of counter $j+1>1$, 
say a transition $(q,\ztinstr {j+1},q')$,
we replace the vectors $(1,-1)$ and $(-1,1)$ in the first
and last transition in \eqref{eq:sim3} by
\[
(j,\ldots,j,-j, j, \ldots, j)
\qquad\text{ and } \qquad
(-j,\ldots,-j,j, -j, \ldots, -j),
\]
respectively, where the unique value of negative (left) or positive
(right) sign appears on the $(j+1)$th coordinate.
% \mei{I don't think it should be $d$, I expected $j-1$ and $-(j-1)$
% respectively, to match the special case worked out for 2-\cm{s}.}

\para{Symmetric zero-tests and single-counter updates}
We do not know how much weaker is the fragment with only 
single-counter $\xupdinstr$ transitions (but unrestricted
$\xztinstr$ transitions).
However, imposing both restrictions results in a degeneracy,
namely the model of $d$-\cm{s}, $d\geq 2$,
is simulated by (and equivalent to) 1-\vass,
at least as long as zero source and target vectors are considered.
Indeed, 
whenever a symmetric \cm has a run 
$q(0,\ldots, 0)\step^* q'(0,\ldots, 0)$, it also has a
\emph{simple} run, in which all counter updates are performed on
counter 1, say, while all zero-tests are (trivially) performed on counter 2,
say.
The zero-tests are thus void, and simple runs are essentially
runs of a 1-\vass.

%% file: complexity.tex
\section{Complexity bound for \Cref{thm:main}} \label{sec:complexity}

% \sla{\bf this section is very sketchy and containts a significant gap.
% After discussions with Lukasz, 
% it seems we can bound doubly-exponentially all the steps except for
% the exhaustive search at the end of the proof of \Cref{lem:compbackfor}...
% }

In this section we provide an upper complexity bound for our decision procedure
outlined in \Cref{sec:outline}.
We will measure the size of input instances of the reachability problem 
as the size of the \emph{unary}
representation of a 2-\cm $\mach=(Q,T)$, the source $s=q(\vr w)$ 
and the target $t=q'(\vr w')$, namely
\[
% \mysize \mach = 
n \, = \,
\mysize Q + \mysize{T_\text{zero-test}} \, + \, 
2 \cdot \norm \V \cdot \mysize {T_\text{update}}
\, + \, \norm{\vr w} \, + \, \norm{\vr w'},
\]
% \mei{wait, do we want to include $w$ and $w'$ within
% $\mysize \mach$? Doesn't feel right...}
% \sla{i think we should. I've removed the notation $\mysize \mach$, let's 
% simply use $n$ instead}
where $T_\text{zero-test}$ and $T_\text{update}$ are the sets of zero-test
and update transitions of $\mach$, respectively, and 
$\V=(Q,T_\text{update})$ is the underlying 2-\vass.
% Note that $\mysize \mach = \mysize \V$, namely we do not distinguish
% between the size of $\mach$ and of the underlying 
% 2-\vass $\V$.
Our objective is to provide an elementary upper bound
in terms of $n$, rather than
optimizing the number of exponents, 
and we will be content with a 5-exponential bound.

\para{Bounds for \Cref{lem:targetpump}}

We estimate the values of constants $C$ and 
$D$.
The constant 
% \[
% C\ = \ C_1 (1 + e\norm\V), 
% \qquad \text{where } \quad
% C_1 \ = \ (2 C_0 e_\mathrm{max} + 1)(\ell_\mathrm{max} + 1) 2\norm\V,
% \] 
% \lukasz{
% I think it should be
\[
    C := C_1 (2 + e'_\mathrm{max}\norm\V), 
    \qquad \text{where } \quad
    C_1 = 2\norm\V e'_\mathrm{max} (\ell_\mathrm{max} + 1),
\]
% x}
is polynomial in $n$.
The norm of the effect of every cycle of length at most $e_\text{max}$ 
is bounded by $e_\text{max}\norm{\V}$,
where $e_\mathrm{max}$ is the maximal length $e$ (polynomial in $n$),
and hence
all values $D_\sigma$ appearing in all the three cases
in the proof of \Cref{lem:targetpump} 
are bounded by $2(e_\text{max} \norm{\V})^2$, thus
polynomially in $n$.
Taking as $D$ the least common multiple of all $D_\sigma$ is enough for the
claim of the lemma.
In summary:

\begin{claim} \label{claim:bounds}
  The constants $C$ and $D$ of \Cref{lem:targetpump} are bounded
  polynomially and exponentially in $n$, respectively.
\end{claim}

\newcommand{\poly}[1]{\text{poly}(#1)}
\para{Terminology}
We use the $k$-exponential function $\expf k$, defined as
$\expf 0 (m) = m$, and $\expf {k+1} (m) = 2^{\expf k (m)}$.
%  (m) = 2^{2^{2^m}}$.
Whenever a base-period representation of a semilinear set $X$ is clear
from the context, 
% \sla{reprezentacji moze byc oczywiscie wiele, mam nadzieje ze to jest
% zrozumiale}, 
we write
$\norm{X}$ for the maximum of norms of bases
and periods of this representation.
In the sequel we say that a semilinear set $X$ is 
\emph{$k$-exponentially semilinear} if given $\mach$,
one can compute
a base-period representation of $X$
of norm bounded by $\norm X \leq \expf k(\poly n)$.
% in time $\norm X \leq \expf k(\poly n)$, and 
%and all computed bases and periods
%are at most (triply) exponential in $n$.
% Furthermore,
% a mapping $S \mapsto f(S)$ transforming a semilinear set $S$
% into a semilinear set $f(S)$ we call 
% \emph{triply exponentially bounded} if
%   given $\mach$ and a base-period representation of 
%   $S\subseteq Q\times\N$, some 
%   representation of $f(S)$ can be computed in time 
%   ${\norm S}^{O(1)}\cdot \expf 3(\poly n)$,
%   % polynomial in $\norm S$ and triply exponential in $n$, 
%   and 
%   $\norm{f(S)} \leq {\norm S}^{O(1)}\cdot \expf 3(\poly n)$.

%  some base-period description of $f(S)$ is computable in time
%  $m$ times (triply) exponential in $n$, and the norms of 
%  all the computed bases 
%  are larger than the largest norm of a base of $S$ by at most
%  (triply) exponential multiplicative factor in $n$.

\para{Bounds for \Cref{lem:0-+}}

According to the linear-path-scheme decomposition of 
\cite[Thm.~1]{Blondin15}, whenever 
a configuration $t$ is reachable from $s$
in the underlying 2-\vass $\V$ of $\mach$,
$t$ is also reachable from $s$ by a run 
\begin{align} \label{eq:lps2}
\sigma \ \in \ 
t_0 \ c_1^{*} \ t_1 \ c_2^{*} \ \ldots \ c_\ell^{*} \ t_\ell,
\end{align}
where 
$\ell$ is polynomial in $n$, and
$t_0, \ldots, t_\ell$ are paths and $c_1, \ldots, c_\ell$ are cycles of
length polynomial in $n$ (cf.~\eqref{eq:lps}).
As shown in \cite{Blondin15},
once $\ell$, $t_0, \ldots, t_\ell$ and $c_1, \ldots, c_\ell$ are fixed, 
the reachability relation along all runs of the form \eqref{eq:lps2} 
is the union of (the projections of) the solution sets of instances $I$ 
of the integer linear programming problem of size polynomial in $n$, 
and hence effectively semilinear.
Moreover, numbers appearing in $I$ are bounded exponentially in
$n$, and hence by the bound of \cite{Pottier91}
(see also \cite[Prop.~4]{taming}), the norms of bases and periods
of the reachability relation along all runs of the form \eqref{eq:lps2}
are bounded exponentially in $n$.
Summing over all choices of $\ell$, $t_0, \ldots, t_\ell$ and 
$c_1, \ldots, c_\ell$, we conclude that
the reachability relation $\step^*$ of $\V$ is exponentially semilinear.
The same applies to the relation $\reachrel\ \subseteq(Q\times\N)^2$,
which is essentially the intersection of $\step^*$ with a linear set
fixing some coordinates to 0.
Finally, we transfer the same effective bounds to the three subrelations
$\stepp 0, \stepp +, \stepp -$ by noticing that each of them is
the intersection of $\reachrel$ with a semilinear set
whose bases and periods are of norm at most $C+1$, the constant of
\Cref{lem:targetpump}
(which is polynomial in $n$),
and intersection of linear sets in fixed dimension (here 2) 
increases the norms of bases and periods at most polynomially
(see \cite[Thm.~6]{taming}):
\begin{claim} \label{claim:cap}
Semilinear sets $r_1,r_2\subseteq (Q\times\N)^2$ satisfy
$\norm{r_1 \cap r_2} \leq (O(1)+\norm{r_1} \cdot \norm{r_2})^{O(1)}$.
\end{claim}
We have thus shown:

\begin{claim} \label{claim:sl}
  The relations $\stepp 0, \stepp +, \stepp -$ are exponentially semilinear.
%   and the norms of bases and periods of their semilinear descriptions
%   are bounded exponentially in $n$.
\end{claim}

\para{Bound for \Cref{lem:0star}}

By the above effective exponential bounds on bases and periods of 
$\stepp 0$,
we conclude that sets $X^c_{q,q'}$ %\eqref{eq:per1},
 \eqref{eq:per2} 
in the proof of \Cref{lem:0star} have 
base-period representations with at most exponential periods and bases.
The period $N^c_{q,q'}$ and threshold $M^c_{q,q'}$ 
of ultimately periodic presentation of $X^c_{q,q'}$ 
are at most doubly exponential, as they are bounded by 
the least common multiplicity of  bases and periods
of its base-period representation.
So is the least common multiple 
$N$ of all periods of $N^c_{q,q'}$,
across all choices of $c, q$ and $q'$, which is doubly exponential.
% By \Cref{claim:sl}, the bases of $X^c_{q,q'}$ are bounded by $N$ as well.
In consequence, the 1-\cm $\mach'$
constructed in the proof of \Cref{lem:0star} is at most
doubly exponentially large.
This, combined with the exponential blow-up of Parikh's theorem
\cite[Thm.~11]{KopczynskiTo:LICS2010},
yields a triple-exponential bound on the norms of bases and periods of 
the relation $\stepp 0^*$.
As both construction of $\mach'$ and Parikh's theorem are effective,
we have thus shown:

\begin{claim} \label{claim:sl2}
  The relation $\stepp{0}^*$ is 3-exponentially semilinear.
%   and the norms of bases and periods of its semilinear description
%   are bounded triply-exponentially in $n$.
\end{claim}

% \para{Composition of semilinar relations}

% Before we continue the analysis we first
% have to prove that fact 
% that that composition of 
% semilinear relations increase their norms 
% polynomially.

% \begin{lemma}\label{lem:composition}
%     There is a number $d \in \N$ such that
%     for semilinear relations $r_1, r_2 \subseteq (Q\times\N)^2$
%     we have $\norm{r_1 \circ r_2} \leq O((\norm{r_1} \cdot \norm{r_2})^d)$.
% \end{lemma}
% \begin{proof}
%     We denote $r_3 = r_1 \circ r_2$.
%     \Mywlog assume $Q = \{0, \ldots, h\}$.
%     For $i = 1,2,3$ we define $r_i' \subseteq \N^2$ as
%     \[
%         r_i' := \setof{(ah + q, a'h+q') \in \N^2  }{  ((q,a), (q',a')) \in r_i }.
%     \]
%     Observe that $\norm{r_i'} = h \cdot \norm{r_i}$
%     and $r_3' = r_1' \circ r_2'$.
%     We define also 
%     \[
%         r_1'' = \setof{(a,b,i) \in \N^3}{(a,b) \in r_1', i \in \N}, \qquad    
%         r_2'' = \setof{(i,a,b) \in \N^3}{(a,b) \in r_2', i \in \N}, 
%     \]
%     and observe that $r_3'$ is obtained from $r_1'' \cap r_2''$
%     by projecting away the middle coordinate.
%     Furthermore,
%     $\norm{r_1''} = \norm{r_1'}$, $\norm{r_2''} = \norm{r_2'}$.
%     We conclude from \cite[Thm.~6]{taming} that there is a polynomial 
%     $p$ independent of sets $r_1,r_2$, such that 
%     $\norm{r_3'} \leq p(\norm{r_1''}, \norm{r_2''})$.
%     Hence, the lemma follows.
% \end{proof}

\para{Bound for \Cref{lem:compbackfor}}

Composition of two semilinear relations in fixed dimension, 
very much like intersection,
increases the norms of bases and periods at most polynomially:

\begin{claim} \label{claim:circ}
Semilinear relations $r_1,r_2\subseteq (Q\times\N)^2$ satisfy
$\norm{r_1 \circ r_2} \leq (O(1)+\norm{r_1} \cdot \norm{r_2})^{O(1)}$.
\end{claim}
Indeed, ignoring $Q$ for simplicity, \Cref{claim:circ} 
follows immediately from \cite[Thm.~6]{taming}, similarly to 
\Cref{claim:cap}, as
the composition $r_1 \circ r_2$ is the projection to the first and last coordinate
of 
\[
(r_1 \times \N) \ \cap \ (\N\times r_2).
\]
Likewise we get the bound for images and inverse images of relations:
\begin{claim} \label{claim:im}
Semilinear relations $r\subseteq (Q\times\N)^2$ and semilinear sets
$X\subseteq Q\times\N$ satisfy
$\norm{r \circ X}, \norm{X \circ r} \leq 
(O(1)+\norm{r} \cdot \norm{X})^{O(1)}$.
\end{claim}

% Using Claims \ref{claim:sl}--\ref{claim:circ} we conclude
% that the relation
%   $(\stepp - \circ \, \stepp{0}^*)$ is triply exponentially semilinear,
% which allows us to deduce further bounds on bases and periods
% of the (necessarily $D$-saturated) inverse image of the relation.

% \begin{claim}\label{claim:sl3}
% The mapping $S \mapsto (\stepp - \circ \, \stepp{0}^*) \,\circ\, S$
%   is triply exponentially bounded.
% \end{claim}

% For later use we state a bound on the bases and periods 
% of the semilinear image of $(\stepp - \circ \, \stepp{0}^*)$:

% \begin{claim}\label{claim:sl3a}
%     The mapping
%   $S\mapsto S \,\circ\,(\stepp - \circ \, \stepp{0}^*)$
%   is triply exponentially bounded.
% \end{claim}

Concerning the first part of \Cref{lem:compbackfor},
we bound the number of iterations of the transformation
\begin{align}\label{eq:tra}
S \quad \mapsto \quad S \ \cup \ 
    ({\stepp -} \, \circ \, {\stepp{0}^*}) \, \circ \,  S
\end{align}
needed to compute the fixpoint $\Phi_3 \circ \{t\}$, starting from
$X_0 = ({\stepp -} \,\circ\ {\stepp{0}^\ast})  \circ \{t\}$.
Due to Claims \ref{claim:sl2} and \ref{claim:im}, 
the bases of the $D$-saturated 
set $X_0$ are bounded triply exponentially in $n$.
Then every base of every next $D$-saturated set
\[
    X_{i+1} \ := \ X_{i} \ \cup \ 
    ({\stepp -} \,\circ\, {\stepp{0}^\ast}) \, \circ\,  X_{i},
\]
as $X_i \subseteq X_{i+1}$,
is either smaller than some base of $X_i$, or is 
\emph{newly created}, namely
 $\minf{X_{i+1}}(q,r) = m < \infty$ while $\minf{X_i}(q,r)=\infty$
 for some $q\in Q$ and $r\in\setfromto 0 {D-1}$.
 A newly created base may appear at most $\mysize Q \cdot D$ 
 times (exponentially many times) during the whole iteration.
 Again due to Claims \ref{claim:sl2} and \ref{claim:im},
%  \begin{claim}\label{claim:sl3m}
% The mapping $S \mapsto (\stepp - \circ \, \stepp{0}^*)^m \,\circ\, S$
%   is triply exponentially bounded.
% \end{claim}
 if a new base appears in $X_{i+1}$ then
 \[
    \norm{X_{i+1}} \leq \norm{X_i}^{O(1)} \cdot \expf 3(\poly n),
 \]
%  \lukasz{Probably it would be possible to prove that $X_{i+1}$ is only multiplied 
%  by some triply exponential factor, without the exponent $d$.
%  However, at the moment I don't see how, and the complexity is the same anyway.}
 otherwise 
 $\norm{X_{i+1}} \leq \norm{X_i}$.
%  \sla{bylo $\norm{X_{i+1}} \leq \max\{\norm{X_i}, D\}$.}
 %larger than a base of $X_i$ by
 %at most triply-exponential multiplicative factor, due to 
 %\Cref{claim:sl3}.
In consequence, for every $X_i$ we 
have 
\begin{align} \label{eq:3e}
\norm{X_i} \leq \bigl(\norm{X_0} \cdot 
\expf 3(\poly n)\bigr)^{2^{O(\mysize Q \cdot D)}},
\end{align}
% \sla{to nie jest w formacie triply expoentially bounded}
% for some polynomial $p'$,
which yields a triply exponential bound on $\norm{X_i}$.
Therefore, the total number of iterations of \eqref{eq:tra} 
before 
a fixpoint is reached is at most triply exponential as well.
% Multiplying a triple exponential value exponentially many times
% yields still a triple exponential quantity,  
%therefore all newly created bases are bounded triply exponentially.
%In consequence, the total number of iterations of \eqref{eq:tra} 
%before 
%a fixpoint is reached is at most triply exponential as well.
Our reasoning was not specific to the starting set $\{t\}$,
and we have actually shown:

\begin{claim}\label{claim:sl4}
  Given a $k$-exponentially semilinear set $S$, $k\geq 3$, 
  the set $\Phi_3 \circ S$ is 
  $k$-exponentially semilinear and computable in
  at most $k$-exponentially many iterations of \eqref{eq:tra}.
\end{claim}

% \smallskip

Concerning the second part of \Cref{lem:compbackfor},
we need to bound the number of iterations 
of the transformation
\begin{align}\label{eq:tra2}
S \quad \mapsto \quad S \ \cup \ f(S),
\qquad \text{ where } \quad
f(S) \ = \ S \ 
    \circ \ \bigl({\stepp{0}^*} \ \circ \ \Phi_3 \ \circ \ {\stepp +}\bigr),
\end{align}
needed to compute the
fixpoint $\{s\} \,\circ\, \Phi_1$, starting from
$X_0 =  f\bigl(\{s\}\bigr)$. 
We should also provide an upper bound on computation of $X_0$,
and computation of every iteration
$
X_{i+1}  :=  X_i  \,\cup\, f\bigl(X_i\bigr).
$

\medskip

We start by analyzing the computation of $X_0 = f(\{s\})$.
By \Cref{claim:sl} the set $(\stepp + \circ \ X_{q,r})$
is exponentially semilinear.
Then by \Cref{claim:sl4} 
applied to $k=3$ and 
$S= (\stepp + \circ \ X_{q,r})$, the set 
$S'=\Phi_3 \,\circ\, ({\stepp +}\,\circ \,X_{q,r}) = 
(\Phi_3 \, \circ \, {\stepp +})\,\circ \,X_{q,r}$
is computable in at most triply exponentially many iterations
of the transformation \eqref{eq:tra}, say $m$ iterations.
Then intersection $Y$ of $S'$
with the set $(\{s\} \, \circ \, {\stepp{0}^*})$, 
\[
    Y \ = \ (\{s\} \, \circ \, {\stepp{0}^*}) \ \cap \ 
\bigl((\Phi_3 \, \circ \, {\stepp +})\, \circ \,X_{q,r}\bigr),
\]
has
(relying on Claims \ref{claim:sl2} and \ref{claim:cap})
bases and periods of at most triply exponential norm. 
Assume, \mywlog, that $m=2^\ell$ for some doubly exponential $\ell$, 
an let $R = ({\stepp -} \, \circ \, {\stepp{0}^*})^{\leq m}$.
The relation $R$ is computable by $\ell$-fold self-composition
$T\mapsto T \cup T^2$ starting from $T=({\stepp -} \, \circ \, {\stepp{0}^*}) \cup \text{Id}$,
each self-composition increasing the norm 
polynomially, which implies:
\begin{claim} \label{claim:sl5}
  The relation $R$ is 4-exponentially semilinear.
%   and the norms of bases and periods of its semilinear description
%   are bounded triply-exponentially in $n$.
\end{claim}
We have $\Phi_3 \circ ({\stepp +}\circ X_{q,r}) =
R \circ ({\stepp +}\circ X_{q,r})$.
Now, a key observation is that once $Y$ is nonempty,
there is a pair $(y,y')\in Y\times ({\stepp +} \circ X_{q,r})$
such that
$ (y,y') \in R$,
which implies that the set
\[ Y \ \circ \ R \ \circ \ {\stepp +}
\]
has necessarily nonempty intersection with $X_{q,r}$, and moreover
the intersection contains, due to Claims \ref{claim:sl}, \ref{claim:circ}, 
\ref{claim:im} and \ref{claim:sl5}, 
a number
which is at most 4-exponential in $n$.
% \sla{nie, wzrost jest wiekszy niz triply exponential; trzeba wzmocnić 
% \Cref{claim:sl3a}}
This provides a 4-exponential upper bound on the number of iterations 
of the exhaustive search 
in the proof of \Cref{lem:compbackfor}, and thus also the same
% , relying on \Cref{claim:sl4}, the 
bound on $\norm{X_0}$.
Testing candidates in the exhaustive search requires only
4-exponentially many iterations of \eqref{eq:tra}, due to
\Cref{claim:sl4} (applied to $k=4$).

When computing the iteration $S \mapsto S \cup f(S)$,
we only apply \Cref{claim:sl4} to sets of the form
${\stepp +} \circ X_{q,r}$, which do not themselves depend
on the input set $S$. As such,
the bound $m$ on the number of iterations does not depend
on the input set $S$ either,
and in the same way as above
we obtain the same 4-exponential bound on the norm of every
subsequent iterate $X_{i+1}=X_i \cup f(X_i)$.

Finally, the 4-exponential bound on the number of iterations 
\eqref{eq:tra2} before a fixpoint is reached in shown similarly
as in the first part.
Similarly as in \eqref{eq:3e} we have
\[
\norm{X_i} \leq \bigl(\norm{X_0} \cdot 
\expf 4(\poly n)\bigr)^{2^{O(\mysize Q \cdot D)}},
\]
which yields a 4-exponential bound on $\norm{X_i}$.
This completes the complexity analysis of the second part of 
\Cref{lem:compbackfor}.

\medskip

% Applying the pushdown automata saturation procedure of \cite{BEM97} 
% to $\mach'$, we conclude that
% the bases and period of $\stepp 0 \circ S$ and $S \circ \stepp 0$
% can be larger than 
% the corresponding quantities of $S$ by at a value at most 
% doubly exponential in $n$.
% Also, composition with $\stepp -$ or $\stepp +$ cannot increase
% bases or periods by larger quantity, due to the exponential bounds
% on bases and periods of semilinear descriptions given by combination
% of \cite{Blondin15} and \cite{Pottier91}.
% We apply these observations to bound the bases of $D$-\saturated sets 
% $X_i$ arising in the iterative computation in the proof of 
% \Cref{lem:compbackfor}.
% In every iteration, either a base is lowered, or a new base appears,
% namely $\minf{X_{i+1}}(q,r) = m < \infty$ while $\minf{X_i}(q,r)=\infty$.
% Whenever this happens, due to the above observations
% the new value $m$ is larger
% than the maximal base in $X_i$ by at most doubly exponential value.
% As this happens at most $\mysize Q D$ times, 
% the largest base ever appearing in any of $X_i$ is at most
% doubly exponential.

\medskip

To summarize, we have obtained the bounds on the norms of 
representations of computed semilinear sets, and on the numbers
of iterations performed in the course of computation.
The bounds are 4-exponential in $n$, the size of unary representation,
and hence at 5-exponential in the size of binary representation.
Combined with an exponential blowup incurred by a necessity to generate
all bases and periods, this yields
the elementary running time of our decision procedure.
As the only known lower bound is \pspace-hardness
(in binary representation),
narrowing the apparent complexity gap is left as future research.